\documentclass[12pt,a4paper]{article}
\usepackage{amsmath,amssymb,amsfonts,amsthm,enumitem}
\newtheorem{theorem}{Theorem}
\newtheorem{lemma}{Lemma}
 
\newtheorem{corollary}{Corollary}
\usepackage[margin=2cm]{geometry}
\usepackage{algorithm,algorithmic}
\usepackage{graphicx}
\usepackage{booktabs}
\usepackage{setspace}
\usepackage{xcolor}
\usepackage{tikz}

\usepackage[round,sort,comma]{natbib}
\usepackage[colorlinks=true,linkcolor=red,citecolor=red,urlcolor=blue]{hyperref}
\usepackage{orcidlink} 

\DeclareMathOperator{\Var}{Var}

\DeclareMathOperator{\sgn}{sgn}

\title{\textbf{Interpolated Quantile Estimation: A Unified Framework Bridging Quantiles and the Mean}}
\usepackage{authblk}

\author[1]{Sa\"{i}d Maanan\,\orcidlink{0000-0001-9365-7586}
\thanks{Corresponding author: \texttt{maanan.said@gmail.com}}}
\author[2]{Azzouz Dermoune}
\author[1]{Ahmed El Ghini}

\affil[1]{LEAM, Mohammed V University in Rabat, Rabat, Morocco}
\affil[2]{Laboratoire Paul Painlevé - CNRS-UMR 8524, Université de Lille, Lille, France}

\date{}

\begin{document}
\maketitle
\begin{abstract}
This paper develops and analyses three families of estimators that continuously interpolate between classical quantiles and the sample mean. The construction begins with an interpolated version of the $L_{1}$ loss, indexed by a location parameter $z$ and an interpolation parameter $h \ge 0$, whose minimizer $\hat q(z,h)$ yields a unified M--estimation framework. Depending on how $(z, h)$ is specified, this framework generates three distinct classes of estimators: fixed-parameter interpolated quantile estimators, plug--in estimators of fixed quantiles, and a new continuum of mean--estimating procedures.

For all three families we establish consistency and asymptotic normality via a uniform asymptotic equicontinuity argument. The limiting variances admit closed forms, allowing a transparent comparison of efficiency across families and interpolation levels. A geometric decomposition of the parameter space shows that, for fixed quantile level $\tau$, admissible pairs $(z, h)$ lie on straight lines along which the estimator targets the same population quantile while its asymptotic variance evolves.

The theoretical analysis reveals two efficiency regimes. Under light--tailed distributions (e.g., Gaussian), interpolation yields a monotone variance reduction. Under heavy--tailed distributions (e.g., Laplace), a finite interpolation parameter $h^{*}(\tau) > 0$ strictly improves efficiency for quantile estimation. Numerical experiments---based on simulated data and real financial returns---validate these conclusions and show that, both asymptotically and in finite samples, the mean--estimating family does not improve upon the sample mean.

\end{abstract}
\clearpage
\section{Introduction}

The classical trade-off between efficiency and robustness is a fundamental theme in statistics, dating back to the influential work of \citet{Huber1981} and \citet{Hampel2005}. 
Two canonical estimators illustrate this contrast:
the sample mean, efficient under light-tailed distributions but highly sensitive to outliers, and sample quantiles, which remain robust under heavy-tailed or contaminated data but may exhibit reduced efficiency under Gaussian-like models.

Let \(Y\) be a real-valued random variable with mean \(m=\mathbb{E}[Y]\), finite variance,
and continuous distribution function \(F\) with density \(f\).
Given observations \(Y_1,\dots,Y_n\), the sample mean
\[
\bar Y = \frac{1}{n}\sum_{i=1}^n Y_i
\]
is the minimizer of the empirical quadratic loss 
\[
q \mapsto \frac{1}{2n}\sum_{i=1}^n (q - Y_i)^2.
\]
In contrast, the sample quantile of order \(\tau\in(0,1)\),
\[
\hat q(\tau)=\inf\{y:\hat F(y)\ge \tau\},
\]
minimizes the empirical asymmetric absolute deviation:
\[
q \mapsto \hat M(q;z)
= \frac{1}{n}\sum_{i=1}^n m(q-Y_i;z),
\qquad
m(q-Y_i;z)=|q-Y_i| + z(q-Y_i),
\]
with the correspondence \( \tau = (1-z)/2 \) \citep{Koenker1978}.
Thus, means and quantiles arise from two classical loss functions, quadratic and absolute, leading to complementary strengths and weaknesses.

\medskip

\paragraph{Interpolation between quantiles and the mean.}
Motivated by the Gibbs-measure formulation of \citet{Dermoune2017} and by regularization ideas in statistical learning 
\citep{Tibshirani1996,Zou2006},
we introduce the following two-parameter objective function:
\[
\hat M(q;z,h)
= \frac{1}{n}\sum_{i=1}^n m(q-Y_i;z,h),
\qquad
m(q-Y_i;z,h)
= |q-Y_i| + z(q-Y_i) + \frac{h}{2}(q-Y_i)^2,
\]
where \(z\in(-1,1)\) and \(h\ge0\).
We denote by
\[
\hat q(z,h)=\arg\min_{q\in\mathbb{R}} \hat M(q;z,h)
\]
the corresponding estimator.  
When \(h=0\), we recover the classical sample quantile of order \(\tau=(1-z)/2\);  
as \(h\to\infty\), for any fixed \(z\), the minimizer converges to the sample mean \(\bar Y\).
This family therefore provides a continuous interpolation between robust quantile estimation and efficient mean estimation.

Beyond this interpolation, the estimator exhibits structural similarities with
trimmed means \citep{Huber2009}, expectiles \citep{Newey1987}, extremiles \citep{Daouia2018},
and more general quantile-like functionals \citep{Bellini2014,Waltrup2014}.
However, the role of the interpolation parameter \(h\) is fundamentally different:
it controls a controlled deformation of the quantile loss toward the quadratic loss, thereby producing a tunable efficiency–robustness balance.

\medskip

\paragraph{Population characterization.}
Let \(q(z,h)\) denote the population minimizer of
\(
M(q;z,h)=\mathbb{E}[m(q-Y;z,h)].
\)
A direct calculation shows that \(q(z,h)\) solves
\begin{equation}
F(q(z,h)) + \frac{h}{2} q(z,h)
= \frac{1 - z + h m}{2}.
\label{eq:population-eqn}
\end{equation}
For each fixed \(z\), the mapping \(h\mapsto q(z,h)\) is monotone and connects the population quantile of order \((1-z)/2\) to the mean \(m\), while for fixed \(h\ge0\), the map \(z\mapsto q(z,h)\) is strictly decreasing.

Equation \eqref{eq:population-eqn} reveals a simple geometric structure:
for every quantile level \(\tau\), all pairs \((z,h)\) satisfying
\[
F(q(z,h))=\tau
\]
lie on a \emph{straight line} in the \((z,h)\)-plane.  
Specifically, setting \(F(q(z,h))=\tau\) in \eqref{eq:population-eqn} gives
\[
\tau = \frac{1-z + h\,(m-q(\tau))}{2},
\]
or equivalently
\[
z = 1 - 2\tau + h\bigl(m - q(\tau)\bigr).
\]
Thus, for fixed \(\tau\), the admissible pairs form the line
\((z(\tau,h), h)\) with \(h\ge0\) and \(z(\tau,h)\) as above.
Along this line, the target quantile remains fixed at \(q(\tau)\) while the corresponding sample estimator changes its asymptotic variance with \(h\). This geometric structure underlies the efficiency--robustness trade-offs analysed in Section~\ref{sec:efficiency}.

\medskip

\paragraph{The Plug--in estimator of a fixed quantile.}
For a fixed target quantile level \(\tau\in(0,1)\), the pair \((z,h)\) that satisfies \eqref{eq:population-eqn} with \(F(q(z,h))=\tau\) is given by
\begin{equation}
z(\tau,h)=1-2\tau + h\bigl(m - F^{-1}(\tau)\bigr). \label{eq:z-tau-h}
\end{equation}
Replacing population quantities by empirical ones yields
\[
\hat z(\tau,h)=1-2\tau + h(\bar Y - \hat q(\tau)).
\]
This fully implementable estimator adds an additional layer of randomness.
A complete Central Limit Theorem for the Plug--in estimator $\hat q(\hat z(\tau,h),h)$, including the covariance corrections induced by estimating \(m\) and \(q(\tau)\), is established in Appendix~\ref{app:plugin-clt}.

\noindent
In addition, we derive in Appendix~\ref{app:mean-clt} a third Central Limit Theorem, devoted to the mean-estimating family $\{\hat q(\hat z,h)\}$, enabling a direct comparison with the sample mean $\bar Y$.

\medskip

\paragraph{Estimating the population mean through the new family.}

A basic property of the population mapping $(z,h)\mapsto q(z,h)$ is that,
for each fixed $z\in(-1,1)$,
\[
q(z,h)\;\longrightarrow\; m \qquad \text{as } h\to\infty.
\]
Hence, for large $h$, the estimator $\hat q(z,h)$ provides a consistent estimator
of the population mean $m$.
This gives a first link between our framework and mean estimation.

A second, fully data-driven construction arises from targeting the value
$z_m = 1 - 2F(m)$, for which $q(z_m,h)=m$ for all $h\ge 0$.
Replacing population quantities by empirical counterparts yields
\[
\hat z = 1 - 2\hat F(\bar Y) + h\bigl(\bar Y - \hat q(\hat F(\bar Y))\bigr).
\]
For any $h\ge 0$, the estimator
\[
\hat q(\hat z,h)
\]
is then a consistent estimator of the population mean $m$.
When $h=0$, this reduces to the sample quantile of order $\hat F(\bar Y)$,
namely $\hat q(\hat F(\bar Y))$.

This construction enlarges the interpretability of our family:
\[
\{\hat q(\hat z,h) : h\ge 0\}
\]
constitutes a continuum of consistent estimators of $m$, interpolating
between the empirical quantile at level $\hat F(\bar Y)$ ($h=0$) and the
sample mean $\bar Y$ (as $h\to\infty$).
A dedicated Central Limit Theorem for this third family of estimators is
established in Appendix~\ref{app:mean-clt}, showing that $\hat q(\hat z,h)$ is
first-order equivalent to $\bar Y$ for every $h\ge 0$.

\medskip

\paragraph{Contributions.}
The main contributions of this paper are:
\begin{enumerate}[label=(\roman*)]
\item A unified two-parameter family \(\hat q(z,h)\) interpolating continuously between quantiles and the sample mean.
\item A complete population characterization, including the linear geometric structure of constant-quantile sets \((z,h)\).
\item Consistency and asymptotic normality for each fixed pair \((z,h)\), with explicit variance expressions.
\item An efficiency analysis showing monotone variance reduction under Gaussian models and non-monotone behaviour, with a finite optimal interpolation level \(h^*(\tau)>0\), under Laplace models.
\item A full Plug--in theory for estimating a fixed quantile: definition, consistency, and a corrected Central Limit Theorem.
\item Numerical experiments, with simulated and real data, validating the theoretical predictions.
\end{enumerate}

The remainder of the paper develops the theory of the interpolated estimator family and evaluates its empirical performance.
Section \ref{sec:family} introduces the unified two--parameter family and its basic properties.
Section \ref{sec:consistency} establishes consistency.
Section \ref{sec:behavior} studies the monotonicity properties of the population mapping and provides practical guidance for parameter selection.
Section \ref{sec:asymptotics} presents the three central limit theorems corresponding to the three estimator classes.
Section \ref{sec:geometry} analyses the geometric structure of the parameter space and the resulting efficiency trade-offs.
Section \ref{sec:numerical_validation} assembles all numerical experiments for Classes (2) and (3).
Section \ref{sec:conclusion} concludes.

\section{A Family of Interpolated Quantile Estimators}
\label{sec:family}

This section introduces the unified two–parameter family of estimators \(\hat q(z,h)\) that underlies the three classes analysed in the next sections.
The construction highlights how a simple quadratic regularization of the classical quantile loss generates a continuum of M-estimators, from the standard sample quantile \((h=0)\) to the sample mean \((h\to\infty)\).
This formulation provides the common framework needed for the asymptotic theory developed in Sections \ref{sec:consistency}–\ref{sec:asymptotics}.

\subsection{Classical Quantiles as Minimizers of a Convex Loss}

Quantiles admit a well-known variational characterization as minimizers of a convex,
piecewise-linear objective. Let $\tau\in(0,1)$ be a fixed quantile level. The sample
quantile
\[
\hat q(\tau)=\inf\{y:\hat F(y)\ge \tau\}
\]
is the unique minimizer (under mild regularity) of the empirical loss
\[
q\mapsto \hat M(q;z)=\frac{1}{n}\sum_{i=1}^n m(q-Y_i;z),
\qquad
m(u;z)=|u|+z\,u,
\]
with the correspondence $\tau=(1-z)/2$ \citep{Koenker1978}. This formulation places
quantiles squarely within the general theory of M-estimation and offers a convenient
starting point for regularization-based extensions.

Under standard regularity assumptions \citep{Koenker1978}, the sample quantile satisfies the classical asymptotic expansion
\[
\sqrt{n}\bigl(\hat q(\tau)-q(\tau)\bigr)
\;\xrightarrow{d}\;
\mathcal N\!\left(0,\frac{\tau(1-\tau)}{f(q(\tau))^2}\right),
\]
where $q(\tau)=F^{-1}(\tau)$ denotes the population quantile. This limit theory will serve
as a benchmark when analysing the behaviour of the interpolated family $\hat q(z,h)$.

\subsection{Quadratic Regularisation as an Interpolation Device}
To construct an estimator interpolating continuously between quantiles and the mean,
we introduce a quadratic regularization into the objective function:
\[
\hat M(q;z,h)
= \frac{1}{n}\sum_{i=1}^n m(q-Y_i;z,h),
\qquad
m(u;z,h) = |u| + z u + \frac{h}{2}u^2,
\]
with $h\ge0$ and where $u = q - Y_i$.
Differentiating the objective yields the empirical score function:
\[
\frac{\partial}{\partial q} \hat M(q;z,h)
= \frac{1}{n}\sum_{i=1}^n \left[\operatorname{sgn}(q-Y_i) + z + h(q-Y_i)\right]
= \hat\Psi(q;z,h),
\]
confirming that our definition of $m$ is consistent with the canonical score function.
\subsection{Canonical Definition of the Score Function}
\label{subsec:score-definition}
We establish a fixed convention for the score function that will be used throughout the paper.

For parameters $z \in (-1,1)$ and $h \ge 0$, define
\[
\psi(t;z,h) = \operatorname{sgn}(t) + h t + z, \qquad t \in \mathbb{R}.
\]
\textbf{Consistency:} With this definition, the loss function whose derivative is $\psi$ is
\[
m(u;z,h) = |u| + z u + \frac{h}{2}u^2,
\]
so that $\frac{\partial}{\partial q} m(q-Y_i;z,h) = \psi(q-Y_i;z,h)$.
The empirical score function is then
\[
\hat\Psi(q;z,h) = \frac{1}{n}\sum_{i=1}^n \psi(q-Y_i;z,h) = \frac{1}{n}\sum_{i=1}^n \bigl[\operatorname{sgn}(q-Y_i) + h(q-Y_i) + z\bigr].
\]
Using the identity $\operatorname{sgn}(q-Y_i) = 1 - 2\cdot\mathbf{1}_{\{Y_i \le q\}}$, we obtain the equivalent but often more convenient form
\begin{equation}
\hat\Psi(q;z,h) = 2\hat F(q) - 1 + z + h(q - \bar Y).
\label{eq:empirical-score}
\end{equation}
The population score function is
\[
\Psi(q;z,h) = \mathbb{E}[\psi(q-Y;z,h)] = 2F(q) - 1 + z + h(q - m).
\]
All subsequent expressions involving $\psi$ or $\Psi$ adhere to these definitions.

Three classes of estimators are derived from this framework:
\begin{itemize}
\item \textbf{Class (1): Fixed parameters.} For any fixed pair $(z,h)$ with $z\in(-1,1)$ and $h\ge 0$, 
      $\hat q(z,h)$ provides an interpolated quantile estimator. When $h=0$, this reduces to the classical 
      sample quantile of order $\tau=(1-z)/2$.

\item \textbf{Class (2): Plug--in for fixed quantile.} For a target quantile level $\tau$, 
      the data-driven parameter $\hat z(\tau,h)=1-2\tau+h(\bar Y-\hat q(\tau))$ yields the 
      estimator $\hat q(\hat z(\tau,h),h)$, targeting $q(\tau)$.

\item \textbf{Class (3): Mean-estimating family.} Using $\hat z = 1-2\hat F(\bar Y)+h(\bar Y-\hat q(\hat F(\bar Y)))$, 
      the estimator $\hat q(\hat z,h)$ targets the population mean $m$.
\end{itemize}

\paragraph{Connection to related work.}
Thus, the parameter \(h\) governs a continuous transition between robustness (small \(h\)) 
and efficiency (large \(h\)), producing a continuum of estimators between the median 
and the mean. This idea parallels regularization principles used in statistical 
learning \citep{Tibshirani1996, Zou2006}, but with a novel focus on the 
quantile–mean trade-off.

\noindent
Several alternative approaches have been developed to interpolate between the 
classical $L_1$ and $L_2$ loss functions. In particular, \citet{Newey1987} introduced 
\emph{asymmetric least squares} estimation, leading to \emph{expectiles}, which 
continuously interpolate between the mean and quantiles depending on an asymmetry 
parameter. Further developments include \emph{extremiles} \citep{Daouia2018}, 
which extend expectiles toward heavy-tailed regimes, and other families of 
generalized quantile functionals \citep{Bellini2014,Waltrup2014}. These approaches 
highlight the continuous spectrum between robustness and efficiency obtained by 
varying asymmetry or tail-sensitivity parameters.

\noindent
Our framework differs in that the interpolation is driven by a 
\emph{interpolation parameter} \(h\) applied directly to the loss, rather than by 
asymmetric weighting of residuals. This yields a distinct yet complementary 
mechanism for transitioning between quantile-based and mean-based estimation.

\noindent
In the framework of robust statistics, both the sample mean and sample quantiles can 
be viewed as special cases of M-estimators associated with different loss functions 
\citep[see, e.g.,][]{Huber2009}. Building on this connection, our interpolated family 
provides a unified M-estimation formulation retaining the robustness properties of 
quantiles while gaining efficiency through quadratic regularization.

\subsection{Implementable Estimator via Plug--in}

For Class (2), the theoretical mapping \(q(z(\tau,h),h)\) depends on unknown population quantities.
Replacing these by their empirical counterparts yields the implementable Plug--in estimator
\[
\hat q(\hat z(\tau,h),h),
\]
whose asymptotic behaviour includes additional covariance terms arising from the estimation of \(\bar Y\) and \(\hat q(\tau)\).
Its central limit theorem appears in Section \ref{sec:asymptotics}.

\section{Consistency} 
\label{sec:consistency}

We now establish consistency for the full two-parameter family \(\hat q(z,h)\).
Since all three estimator classes arise as specialisations of this family, a single uniform consistency result suffices for the entire framework.
From the law of large numbers we have almost surely 
\(\hat M(q;z,h)\to M(q;z,h)\) and \(\hat \Psi(q;z,h)\to \Psi(q;z,h)\)
for each \(q, z\in\mathbb{R}\) and \(h\geq 0\) 
as \(n\to +\infty\), with 
\[
M(q;z,h)=\mathbb{E}[m(q-Y;z,h)], \qquad 
\Psi(q;z,h)=\mathbb{E}[\psi(q-Y;z,h)].
\]
The minimizer \(q(z,h)\) of \(q\mapsto M(q;z,h)\) satisfies the first-order 
condition \(\Psi(q;z,h)=\mathbb{E}[\psi(q-Y;z,h)]=0\). 
A similar calculation as in a finite sample shows that 
\[
q(z,h)=\left(F+\tfrac{h}{2}I_{\mathbb{R}}\right)^{-1}
\!\left(\tfrac{1-z+hm}{2}\right),
\]
or equivalently
\[
F(q(z,h))+\tfrac{h}{2}q(z,h)=\tfrac{1-z+hm}{2}.
\]
If \(z\in (-1,1)\), then 
\(h\mapsto q(z,h)\) varies between \(q(z,0)\), the population quantile of order 
\(\tfrac{1-z}{2}\), and the population mean \(m\).

As \(q\mapsto \hat \Psi(q;z,h)\) is nondecreasing, standard M-estimation
arguments ensure that the sample minimizer converges almost surely:
\[
\hat q(z,h)\to q(z,h) \quad \text{a.s.}
\qquad \text{\citep[Lemma~5.10]{Vaart1998}.}
\]

\medskip
\noindent
We next investigate in detail how the population function \(q(z,h)\)
varies with respect to its parameters \(z\) and \(h\),
and how the interpolation parameter \(h\) drives the transition from robustness
to efficiency.

\section{Behaviour of \texorpdfstring{$q(z,h)$}{q(z,h)} and Some Practical Implications}
\label{sec:behavior}
Before turning to asymptotic normality, we study the population mapping \(q(z,h)\).
Its monotonicity properties with respect to \(z\) and \(h\) determine both the geometry of the parameter space and the variance–bias trade-offs appearing in the CLTs and efficiency results that follow.
Recall that $q(z,h)$ is defined implicitly by the first-order condition
\[
F(q(z,h)) + \frac{h}{2} q(z,h) = \frac{1 - z + h m}{2},
\]
where $F$ denotes the cumulative distribution function of $Y$, and $m = \mathbb{E}[Y]$ is the population mean.

Differentiating both sides of this identity with respect to $h$ yields
\[
\frac{\partial q}{\partial h}(z,h)
= \frac{m - q(z,h)}{2f(q(z,h)) + h},
\]
where $f$ is the probability density function of $Y$.
It follows that for each fixed $z \in (-1,1)$:
\begin{itemize}
  \item If $q(z,0) < m$, then $h \mapsto q(z,h)$ increases monotonically from $q(z,0)$ to $m$.
  \item If $q(z,0) > m$, then $h \mapsto q(z,h)$ decreases monotonically from $q(z,0)$ to $m$.
\end{itemize}
Hence, as $h$ grows, the population quantity $q(z,h)$ continuously moves toward the mean $m$,
providing a continuous interpolation between quantile-based estimation ($h=0$) and mean-based estimation ($h \to \infty$).

\medskip
\noindent\textbf{Empirical counterpart.}
The same monotonicity properties hold for the empirical estimator $\hat q(z,h)$: 
for each fixed $z$, as $h$ increases from $0$, 
\begin{itemize}
    \item if $\hat q(z,0) < \bar Y$, then $h \mapsto \hat q(z,h)$ increases monotonically 
          from $\hat q(z,0)$ toward $\bar Y$;
    \item if $\hat q(z,0) > \bar Y$, then $h \mapsto \hat q(z,h)$ decreases monotonically 
          from $\hat q(z,0)$ toward $\bar Y$.
\end{itemize}
This empirical preservation of monotonicity is verified numerically in 
Section~\ref{subsec:monotonicity} (see Figure~\ref{fig:monotonicity}) and 
in the real-data application of Section~\ref{subsec:realdata}.

\medskip
\noindent
Differentiating instead with respect to $z$ gives
\[
\frac{\partial q}{\partial z}(z,h)
= -\frac{1}{2f(q(z,h)) + h}.
\]
This implies that for each fixed $h>0$, the function $z \mapsto q(z,h)$ is strictly decreasing,
with $q(z,h) \to +\infty$ as $z \downarrow -1$ and $q(z,h) \to -\infty$ as $z \uparrow 1$.
Hence, it is sufficient to consider $z \in (-1,1)$.

\medskip
\noindent\textbf{Practical implications.} 
As a consequence of these results we show how to estimate quantiles using 
our estimators $\hat q(z,h)$:
\begin{itemize}
    \item If we want to estimate a quantile with small order we choose 
          $z$ and $h$ close respectively to $-1$ and $0$.
    \item If we want to estimate a quantile close to the mean 
          we choose $\hat q((1-z)/2) < \bar Y$ and $h$ large.
    \item If we want to estimate a quantile above the mean we choose 
          $\hat q((1-z)/2) > \bar Y$.
\end{itemize}

\section{Asymptotic Normality}
\label{sec:asymptotics}
We now derive the three central limit theorems corresponding to the three estimator classes introduced in Section \ref{sec:family}:
\begin{enumerate}[label=(\roman*)]
\item the interpolated quantile estimators $\hat q(z,h)$;
\item the plug--in estimators of a fixed quantile $\hat q(\hat z(\tau,h),h)$;
\item the mean--estimating family $\hat q(\hat z(h),h)$, where $\hat z(h)=1-2\hat F(\bar Y)+h\bigl(\bar Y- \hat q(\hat F(\bar Y))\bigr)$.
\end{enumerate}
The parameter \(h\) is always considered fixed (not shrinking with \(n\)).
All proofs, which rely on uniform asymptotic equicontinuity of the empirical process,
are postponed to the appendices.

\subsection*{Regularity Assumptions}

Throughout this section we impose the following conditions:
\begin{itemize}
  \item[(A1)] $F$ is continuous and $Y$ has a finite second moment;
  \item[(A2)] $Y$ admits a continuous, strictly positive density $f$ in a neighbourhood
  of all relevant points.
\end{itemize}

These conditions ensure the validity of Knight-type expansions, Z--estimator
linearisation, and the functional delta method used for the plug--in corrections.

\subsection{CLT for the Interpolated Quantile Estimator}
\label{subsec:clt1}

Consider the two--parameter estimator family
\[
\hat q(z,h)
=\arg\min_{q\in\mathbb{R}} \hat M(q;z,h),
\qquad
z\in(-1,1),\; h\ge 0,
\]
with population counterpart $q(z,h)$ defined by
$\Psi(q(z,h);z,h)=0$, where $\Psi$ is the population score.
\begin{theorem}[CLT for the interpolated quantile estimator]
\label{thm:CLT}
Under Assumptions (A1)--(A2),
\[
\sqrt{n}\bigl(\hat q(z,h)-q(z,h)\bigr)
\;\xrightarrow{d}\;
\mathcal{N}\!\left( 0,\;\frac{B(z,h)}{\bigl( 2f(q(z,h))+h \bigr)^2} \right),
\]
where
\[
B(z,h)
=4F(q(z,h))(1-F(q(z,h)))
+2h\Bigl[\mathbb{E}|Y-q(z,h)|
 -(m-q(z,h))(1-2F(q(z,h)))\Bigr]
+h^2\Var(Y).
\]
\end{theorem}

The proof is given in Appendix~\ref{app:clt1} and is based on uniform
asymptotic equicontinuity of the empirical process
$\sqrt{n}\,(\hat F - F)$ combined with Knight’s identity.
Uniform asymptotic equicontinuity ensures that replacing \(q(z,h)\) by its estimator \(\hat q(z,h)\) inside the indicator process produces an \(o_p(1)\) remainder uniformly over bounded neighbourhoods.

\medskip
\noindent\textbf{Remark.}
When \(h = 0\), the asymptotic variance reduces to \(\tau(1-\tau)/f(q(\tau))^2\), 
the classical variance of the sample quantile. For \(h > 0\), the quadratic 
regularization introduces additional terms that blend quantile-like and mean-like 
behaviour, leading to the variance expression \(B(z,h)/(2f(q(z,h))+h)^2\). 
This hybrid structure enables the efficiency--robustness trade-off analysed in 
Section~\ref{sec:efficiency}.

\begin{corollary}[Variance along constant-quantile lines]
\label{cor:hstar}
For a fixed quantile level $\tau\in(0,1)$, consider the parameter pairs $(z(\tau,h),h)$ that satisfy $F(q(z(\tau,h),h))=\tau$, where
\[
z(\tau,h)=1-2\tau+h\bigl(m-F^{-1}(\tau)\bigr).
\]
Along this line, the asymptotic variance from Theorem~\ref{thm:CLT} simplifies to
\[
v(\tau,h)=\frac{B(z(\tau,h),h)}{\bigl(2f(F^{-1}(\tau))+h\bigr)^2},
\]
where $B(z,h)$ is defined in Theorem~\ref{thm:CLT} and reduces to
\[
A(\tau,h)=4\tau(1-\tau)+2h\!\left[\mathbb{E}|Y-F^{-1}(\tau)|-(m-F^{-1}(\tau))(1-2\tau)\right]+h^2\Var(Y).
\]

Let
\[
a=4\tau(1-\tau),\quad 
b=2\!\left[\mathbb{E}|Y-F^{-1}(\tau)|-(m-F^{-1}(\tau))(1-2\tau)\right],\quad 
c=\Var(Y),\quad 
d=2f(F^{-1}(\tau)).
\]
Then
\[
\frac{\partial v(\tau,h)}{\partial h}
=\frac{(2cd-b)h+(bd-2a)}{(d+h)^3}.
\]
According to the signs of $bd-2a$ and $2cd-b$, the following cases arise:
\begin{enumerate}
    \item[(a)] If $bd-2a>0$ and $2cd-b\ge0$, then $v(\tau,0)<v(\tau,h)$ for all $h>0$.
    \item[(b)] If $bd-2a<0$ and $2cd-b\ge0$, then $v(\tau,h)<v(\tau,0)$ for all $h>0$.
          In this case, every $h>0$ improves efficiency, and $\arg\min_{h>0}v(\tau,h)=+\infty$.
    \item[(c)] If $(bd-2a)(2cd-b)<0$, there exists $h^*(\tau)>0$ such that
          $v(\tau,h^*(\tau))<v(\tau,0)$, giving a finite optimal interpolation level.
\end{enumerate}
\end{corollary}

\subsection{CLT for the Plug--in Estimator of a Fixed Quantile}
\label{subsec:clt2}

Fix $\tau\in(0,1)$ and $h\ge 0$.  
The corresponding population and empirical parameters are
\[
z(\tau,h)=1-2\tau+h\bigl(m-F^{-1}(\tau)\bigr),
\qquad
\hat z(\tau,h)=1-2\tau+h\,(\bar Y-\hat q(\tau)).
\]
The plug--in estimator of the $\tau$--quantile is
\[
\hat q\bigl(\hat z(\tau,h),h\bigr).
\]
Compared with fixed-$z$ inference, this estimator introduces additional
variability through $(\bar Y,\hat q(\tau))$, whose joint effect is
handled via uniform asymptotic equicontinuity and the functional delta method.
The complete derivation appears in Appendix~\ref{app:plugin-clt}.
\begin{theorem}[CLT for the plug--in estimator of a fixed quantile]
\label{thm:plugin-clt}
Under Assumptions (A1)--(A2), for any fixed \(\tau \in (0,1)\) and \(h \ge 0\),
\[
\sqrt{n}\bigl(\hat q(\hat z(\tau,h),h) - q(\tau)\bigr)
\;\xrightarrow{d}\;
\mathcal{N}\!\left(0,\;\frac{\tau(1-\tau)}{f^2(q(\tau))}\right).
\]
In particular, the plug--in estimator \(\hat q(\hat z(\tau,h),h)\) is first-order 
equivalent to the ordinary sample quantile \(\hat q(\tau)\), and hence shares its 
asymptotic distribution.
\end{theorem}
This corrected variance includes the contribution of estimating $\bar Y$ and
$\hat q(\tau)$, and differs from the population parametrised variance even when
$h=0$.
\subsection{CLT for the Mean--Estimating Family}
\label{subsec:clt3}

Define the data--driven quantity
\[
\hat z = 1 - 2\,\hat F(\bar Y) + h\bigl(\bar Y - \hat q(\hat F(\bar Y))\bigr),
\]
which estimates $z_m = 1 - 2F(m)$, the value for which
$q(z_m,h) = m$ for all $h \ge 0$.
For each fixed $h \ge 0$, consider the estimator
\[
\hat q(\hat z,h).
\]
This yields a continuum of consistent estimators of $m$,
interpolating between the empirical quantile at level $\hat F(\bar Y)$
(when $h = 0$) and the sample mean (as $h \to \infty$).
The proof of the CLT relies on the same Z--estimator arguments and on uniform
asymptotic equicontinuity of the empirical process (Appendix~\ref{app:mean-clt}).

\begin{theorem}[CLT for the mean--estimating family]
\label{thm:mean-clt-main}
Assume (A1)--(A2). For every fixed \(h \ge 0\), define the data--driven level
\[
\hat z = 1 - 2\hat F(\bar Y) + h\bigl(\bar Y - \hat q(\hat F(\bar Y))\bigr),
\]
and consider the estimator \(\hat q(\hat z,h)\).
Then
\[
\sqrt{n}\bigl(\hat q(\hat z,h) - m\bigr)
\;\xrightarrow{d}\;
\mathcal{N}\!\bigl(0,\, \Var(Y)\bigr).
\]
In particular, for each fixed \(h \ge 0\), the estimator \(\hat q(\hat z,h)\) is
first--order equivalent to the sample mean and therefore asymptotically efficient
for estimating \(m\).
\end{theorem}

\noindent\textbf{Remark.}
The result follows from the linear representation of Theorem~\ref{thm:CLT} applied at the random level \(\hat z\), combined with a first--order expansion of \(\hat F(\bar Y)\) around \(F(m)\). The detailed argument is given in Appendix~\ref{app:mean-clt}. The key simplification is that
\[
\hat q(\hat z,h) - m = \bar Y - m + o_p(n^{-1/2}),
\]
showing first--order equivalence to the sample mean for every \(h \ge 0\).

This CLT allows a direct efficiency comparison between the sample mean and its
interpolated competitors $\hat q(\hat z,h)$.
This result is the key theoretical tool for answering the question whether a interpolated estimator \(\hat q(\hat z,h)\) can outperform the sample mean \(\bar Y\).
\section{Parameter Geometry}
\label{sec:geometry}
This section connects the analytical results of Section~\ref{sec:asymptotics} with the geometric structure of the parameter space introduced in Section~\ref{sec:family}. The key observation is that for each fixed quantile level, the admissible parameter pairs lie on straight lines, and moving along these lines has important implications for the three estimator classes.
\subsection{Geometric Structure of the Parameter Space}
For each fixed quantile level $\tau\in(0,1)$,
the parameter pairs $(z,h)$ for which the population solution satisfies
\[
F(q(z,h)) = \tau
\]
lie on a straight line in the $(z,h)$-plane. Specifically, solving the
population equation $F(q(z,h)) + \frac{h}{2} q(z,h) = \frac{1 - z + hm}{2}$ for
$F(q(z,h)) = \tau$ yields the linear constraint
\[
\tau = \frac{1}{2}\big(h(m-F^{-1}(\tau)) + (1-z)\big),
\]
or equivalently
\[
z = 1 - 2\tau + h\bigl(m - F^{-1}(\tau)\bigr).
\]
Thus, for fixed $\tau$, the admissible pairs form the line
$(z(\tau,h), h)$ with $h\ge0$ and $z(\tau,h)$ as above.

\medskip
\noindent
The asymptotic variance along these lines and its behaviour with respect to $h$
are characterized in Corollary~\ref{cor:hstar} (Section~\ref{subsec:clt1}),
which shows that interpolation can improve efficiency for heavy-tailed distributions.
\subsection{Implications for the Three Estimator Classes}
\label{sec:efficiency}
The geometric structure has distinct consequences for each of the three estimator families:
\paragraph{Class (1): Fixed $(z,h)$ estimators.}
For fixed $z\in(-1,1)$ and $h\ge0$, Theorem~\ref{thm:CLT} provides the
asymptotic distribution of $\hat q(z,h)$. As $h$ varies, the estimator
interpolates between the $\tau$-quantile (with $\tau=(1-z)/2$) and the
mean $m$.
\paragraph{Class (2): Plug--in estimator of a fixed quantile.}
For a target quantile level $\tau$, the natural approach would be to use
the pair $(z(\tau,h),h)$ from the line defined above. However, since
$z(\tau,h)$ depends on unknown population quantities $m$ and $F^{-1}(\tau)$,
one replaces them by their empirical counterparts, leading to the
plug--in estimator $\hat q(\hat z(\tau,h),h)$ with
$\hat z(\tau,h) = 1-2\tau + h(\bar Y - \hat q(\tau))$.

Theorem~\ref{thm:plugin-clt} shows that this plug--in construction yields
an estimator that is \emph{first-order equivalent} to the ordinary sample
quantile $\hat q(\tau)$. Consequently, interpolation does not improve the
asymptotic efficiency for estimating a fixed quantile when using the
plug--in approach.
\paragraph{Class (3): Mean--estimating family.}
For mean estimation, the relevant line in parameter space is the one
that passes through the point $(z_m, h)$ for all $h$, where
$z_m = 1-2F(m)$ ensures $q(z_m,h)=m$. Using the empirical counterpart
$\hat z = 1-2\hat F(\bar Y) + h\bigl(\bar Y - \hat q(\hat F(\bar Y))\bigr)$,
Theorem~\ref{thm:mean-clt-main} establishes that $\hat q(\hat z,h)$ is
first-order equivalent to the sample mean $\bar Y$ for every $h\ge0$.
\subsection{Summary of Main Findings}
The geometric viewpoint clarifies the different roles played by the interpolation
parameter $h$ in the three families:
\begin{itemize}
  \item For fixed $(z,h)$, the estimator $\hat q(z,h)$ continuously interpolates
        between a quantile and the mean, with asymptotic variance given by
        Theorem~\ref{thm:CLT}.

  \item For the plug--in quantile estimator, interpolation does not translate into
        asymptotic efficiency gains; the estimator remains equivalent to the
        classical sample quantile.

  \item For mean estimation, the interpolated family $\hat q(\hat z,h)$ is
        asymptotically equivalent to the sample mean for all $h\ge0$, showing
        that interpolation cannot improve upon the efficiency of $\bar Y$.
\end{itemize}
These conclusions are validated numerically in Section~\ref{sec:numerical_validation}.
\section{Numerical Validation of Theoretical Results}
\label{sec:numerical_validation}

\subsection{Experiment 1: Fixed \texorpdfstring{$(z,h)$}{(z,h)} vs Classical Quantile}
\label{subsec:exp-a}

This experiment illustrates the efficiency comparison established in
Theorem~\ref{thm:CLT} between the interpolated estimator
$\hat q(z,h)$ at a fixed pair $(z,h)$ and the classical sample quantile
targeting the \emph{same population quantile}.
For each $(z,h)$, let $q(z,h)$ denote the population solution of the estimating
equation and define $\tau(z,h)=F(q(z,h))$.
The relevant benchmark is therefore the classical quantile estimator
$\hat q(\tau(z,h))$, whose asymptotic variance equals
$\tau(z,h)\bigl(1-\tau(z,h)\bigr)/f^2\!\bigl(q(z,h)\bigr)$.

\paragraph{Design.}
We consider two distributions:
(i) the standard normal distribution and
(ii) the standard Laplace distribution.
For each distribution, we fix three values of $z\in\{-0.5,0,0.5\}$,
corresponding at $h=0$ to the quantile levels $\tau_0=(1-z)/2\in\{0.25,0.5,0.75\}$.
For each fixed $z$, the interpolation parameter $h$ varies on a grid
$h\in[0,5]$.
For every $(z,h)$, the population equation is solved numerically to obtain
$q(z,h)$ and $\tau(z,h)$.

\paragraph{Quantities compared.}
For each $(z,h)$, we compute:
\begin{itemize}
\item the asymptotic variance $\sigma^2(z,h)$ of $\hat q(z,h)$ given by
Theorem~\ref{thm:CLT}, and
\item the classical quantile variance
$\tau(z,h)\bigl(1-\tau(z,h)\bigr)/f^2\!\bigl(q(z,h)\bigr)$.
\end{itemize}
The results are summarized through the ratio
\[
R(z,h)
=\frac{\sigma^2(z,h)}
{\tau(z,h)\bigl(1-\tau(z,h)\bigr)/f^2\!\bigl(q(z,h)\bigr)}.
\]
Values $R(z,h)<1$ indicate that the interpolated estimator $\hat q(z,h)$ is
asymptotically more efficient than the classical quantile estimator
targeting the same population quantile.

\paragraph{Results.}
Figure~\ref{fig:exp-a} reports $R(z,h)$ as a function of $h$ for the three values
of $z$ and for both distributions.
By construction, $R(z,0)=1$ for all $z$, since $h=0$ corresponds to the
ordinary sample quantile.
For $h>0$, the behaviour depends on both the distribution and the value of $z$.

\begin{figure}[!htbp]
\centering
\includegraphics[width=0.95\textwidth]{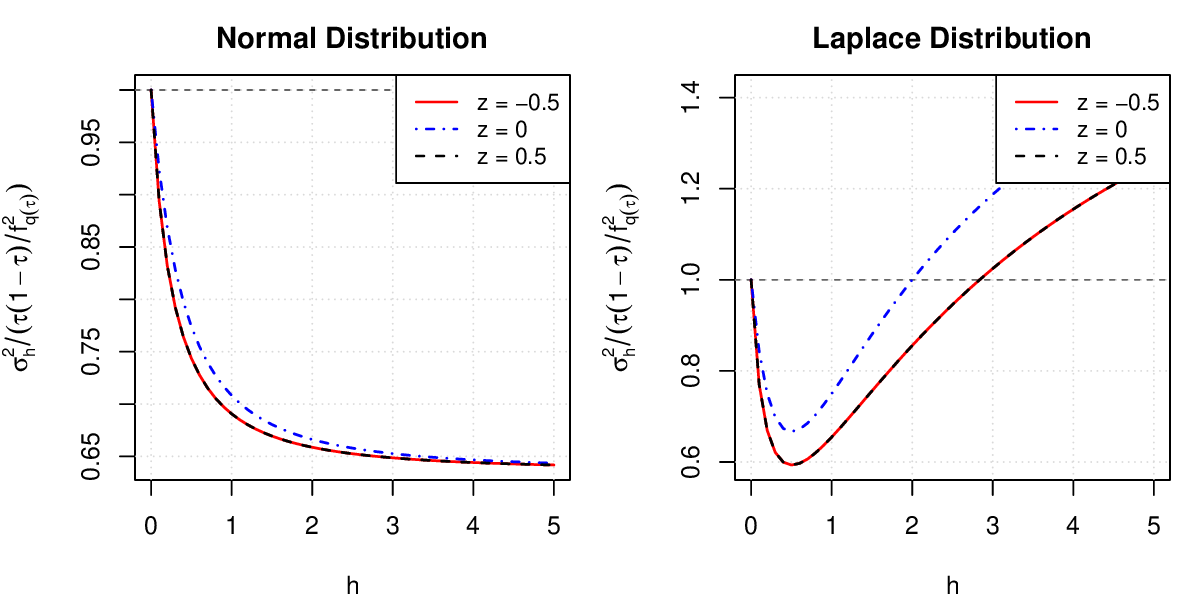}
\caption{Experiment 1: Efficiency ratio $R(z,h)$ for the fixed $(z,h)$ estimator vs. the classical quantile estimator. Left: Normal distribution. Right: Laplace distribution. Three values of $z$ are shown: $z=-0.5$ (solid), $z=0$ (dashed), and $z=0.5$ (dotted). The horizontal line at 1 indicates the classical quantile variance.}
\label{fig:exp-a}
\end{figure}

For the normal distribution, the ratio $R(z,h)$ decreases monotonically in $h$
for all three values of $z$, indicating a uniform efficiency gain from interpolation.
The reduction is particularly pronounced around the median ($z=0$), but is also
clearly visible for upper and lower quartiles.

For the Laplace distribution, the behaviour is non-monotone.
For $z=\pm0.5$, the ratio initially decreases below one, showing a substantial
efficiency gain for moderate values of $h$, before eventually increasing and
exceeding one for larger $h$.
In contrast, for $z=0$ (the median), the ratio first decreases, then increases,
crossing one at a finite value of $h$.
This pattern is fully consistent with the cases described in
Corollary~\ref{cor:hstar}, which predicts the possible existence of a finite
variance-minimizing interpolation level depending on the distribution and the target
quantile.

Overall, this experiment confirms that, when $(z,h)$ is fixed,
the interpolated estimator $\hat q(z,h)$ can substantially outperform the classical
quantile estimator targeting the same population quantile, and that the magnitude
and shape of the efficiency gains depend sensitively on both the underlying
distribution and the quantile level.

\subsection{Experiment 2: Fixed \texorpdfstring{$(z,h)$}{(z,h)} – Efficiency Gain under Heavy Tails}
\label{subsec:exp_b}

This experiment investigates the asymptotic efficiency of the interpolated estimator \(\hat q(z,h)\) for fixed parameters \((z,h)\).  
The novel aspect is to compare its asymptotic variance \(\sigma^2(h)\) with the variance of the sample mean \(\Var(Y)\) for finite \(h\).  
For symmetric distributions with mean equal to the median (e.g., Normal and Laplace), choosing \(z=0\) targets the centre for all \(h\), so a direct variance comparison is meaningful because the bias is zero.
The main finding is that for heavy-tailed distributions, the asymptotic variance $\sigma^2(h)$ can fall strictly below $\Var(Y)$ for moderate values of $h$, yielding an estimator of the centre that is more efficient than the sample mean.

\paragraph{Theoretical background.}
For fixed \((z,h)\), Theorem~\ref{thm:CLT} gives the asymptotic variance
\[
\sigma^2(h)=\frac{B(z,h)}{\bigl(2f(q(z,h))+h\bigr)^2},
\]
with \(B(z,h)\) as defined in the theorem.  
As \(h\to\infty\), \(\hat q(z,h)\) converges to the sample mean, hence
\[
\sigma^2(h)\;\longrightarrow\;\Var(Y),\qquad h\to\infty.
\]
For \(h=0\) we recover the classical quantile variance \(\tau(1-\tau)/f^2(q(\tau))\) with \(\tau=(1-z)/2\).
For comparison, we also recall the asymptotic variance of the plug‑in estimator that targets a fixed quantile \(\tau\) (see Theorem~\ref{thm:plugin-clt}):
\[
\frac{\tau(1-\tau)}{f^2(q(\tau))},
\]
which does not depend on \(h\). When \(h\to\infty\), the induced quantile level \(\tau(z,h)=F(q(z,h))\) converges to \(F(m)\), so the limiting plug‑in variance is
\[
\frac{F(m)\bigl(1-F(m)\bigr)}{f^2\!\bigl(F(m)\bigr)}.
\]
While this limiting quantity is classical, the present experiment focuses on the behaviour of $\sigma^2(h)$ for finite values of $h$, where new phenomena arise.

\paragraph{Design of the experiment.}
We consider the standard Normal and standard Laplace distributions, both symmetric about zero, so \(m=0\) and \(F(m)=1/2\). For each distribution, we compute \(\sigma^2(h)\) for three values of \(z\in\{-0.5,0,0.5\}\), corresponding to initial quantiles below, at, and above the median. In addition to \(\sigma^2(h)\), we plot:
\begin{itemize}
	\item the plug‑in variance \(\tau(z,h)(1-\tau(z,h))/f^2(q(\tau(z,h)))\);
	\item the limiting plug‑in variance at \(F(m)\);
	\item the variance \(\Var(Y)\) of the sample mean.
\end{itemize}
All quantities are evaluated at the population level and plotted as functions of \(h\).

\paragraph{Results.}
Figure~\ref{fig:experiment-b} summarizes the results.  
For the Normal distribution, \(\sigma^2(h)\) decreases monotonically from the quantile variance towards \(\Var(Y)=1\), but appears to remain above 1 for all values of $h$. Hence, interpolation does not improve upon the sample mean for light‑tailed distributions. This indicates that for light-tailed distributions, the sample mean remains asymptotically optimal within this class, and quadratic regularization does not yield efficiency gains.

For the Laplace distribution, the behaviour is different.  
For \(z=0\) (the median case), \(\sigma^2(h)\) initially decreases from the median variance (1) to a minimum of approximately 0.67 at \(h\approx0.4\), which is substantially lower than \(\Var(Y)=2\). As \(h\) increases further, \(\sigma^2(h)\) rises and approaches 2 from below. For \(z=\pm0.5\), the variance also dips below 2 for a range of moderate \(h\).  
Thus, for the Laplace distribution, a finite amount of quadratic regularization produces an estimator of the centre with lower asymptotic variance than both the sample mean and the sample median, a genuine efficiency gain.

\paragraph{Interpretation and novelty.}
The key finding is that for heavy-tailed symmetric distributions, the interpolated estimator $\hat q(0,h)$ can achieve an asymptotic variance strictly smaller than $\Var(Y)$ for a range of finite values of $h$. 
This phenomenon does not occur for classical estimators: the sample mean has fixed variance $\Var(Y)$, while the sample median has variance determined solely by the density at the centre. In contrast, the interpolated estimator exploits a hybrid weighting scheme that simultaneously reduces the influence of extreme observations and stabilizes variability, leading to a strictly improved efficiency in intermediate regimes.

This demonstrates that the parameter $h$ does not merely interpolate between quantile and mean behaviour, but creates a new class of estimators with nontrivial efficiency properties that depend on the underlying distribution.

\begin{figure}[!htbp]
	\centering
	\includegraphics[width=\textwidth]{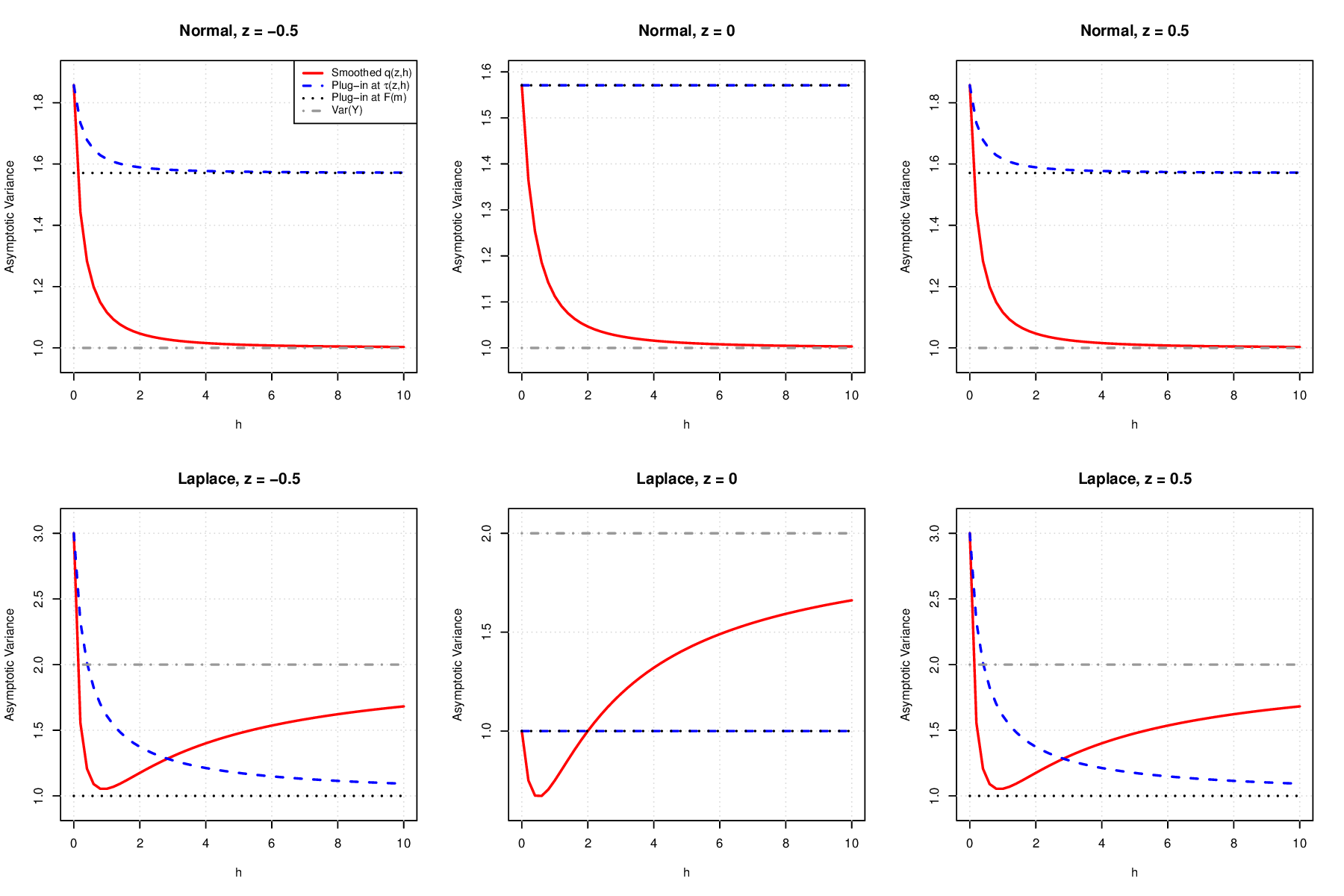}
	\caption{Experiment (b): asymptotic variances of the interpolated estimator \(\hat q_n(z,h)\), the plug‑in estimator at \(\tau(z,h)\), the limiting plug‑in variance at \(F(m)\), and \(\Var(Y)\), as functions of \(h\). Top row: Normal distribution. Bottom row: Laplace distribution. For Laplace and \(z=0\) (centre), the interpolated variance drops well below \(\Var(Y)=2\) (dotted horizontal line), demonstrating a genuine efficiency gain over the sample mean for finite \(h\).}
	\label{fig:experiment-b}
\end{figure}

\paragraph{Numerical illustration for fixed quantile levels.}
To complement Simulation~(b), we provide a finite-sample illustration of the behaviour of the two estimators for fixed quantile levels \(\tau \in \{0.25, 0.5, 0.75\}\).  
For each \(\tau\), we consider the population pair \((z(\tau,h),h)\) such that \(F(q(z(\tau,h),h))=\tau\), and compare the trajectories of \(\hat q(z(\tau,h),h)\) and the plug‑in estimator \(\hat q(\hat z(\tau,h),h)\).  
Figure~\ref{fig:fixed_tau_paths} shows that both estimators remain centred around the target quantile \(q(\tau)\), confirming that the variance differences observed in Figure~\ref{fig:experiment-b} are not due to targeting discrepancies.

\begin{figure}[!htbp]
	\centering
	\includegraphics[width=\textwidth]{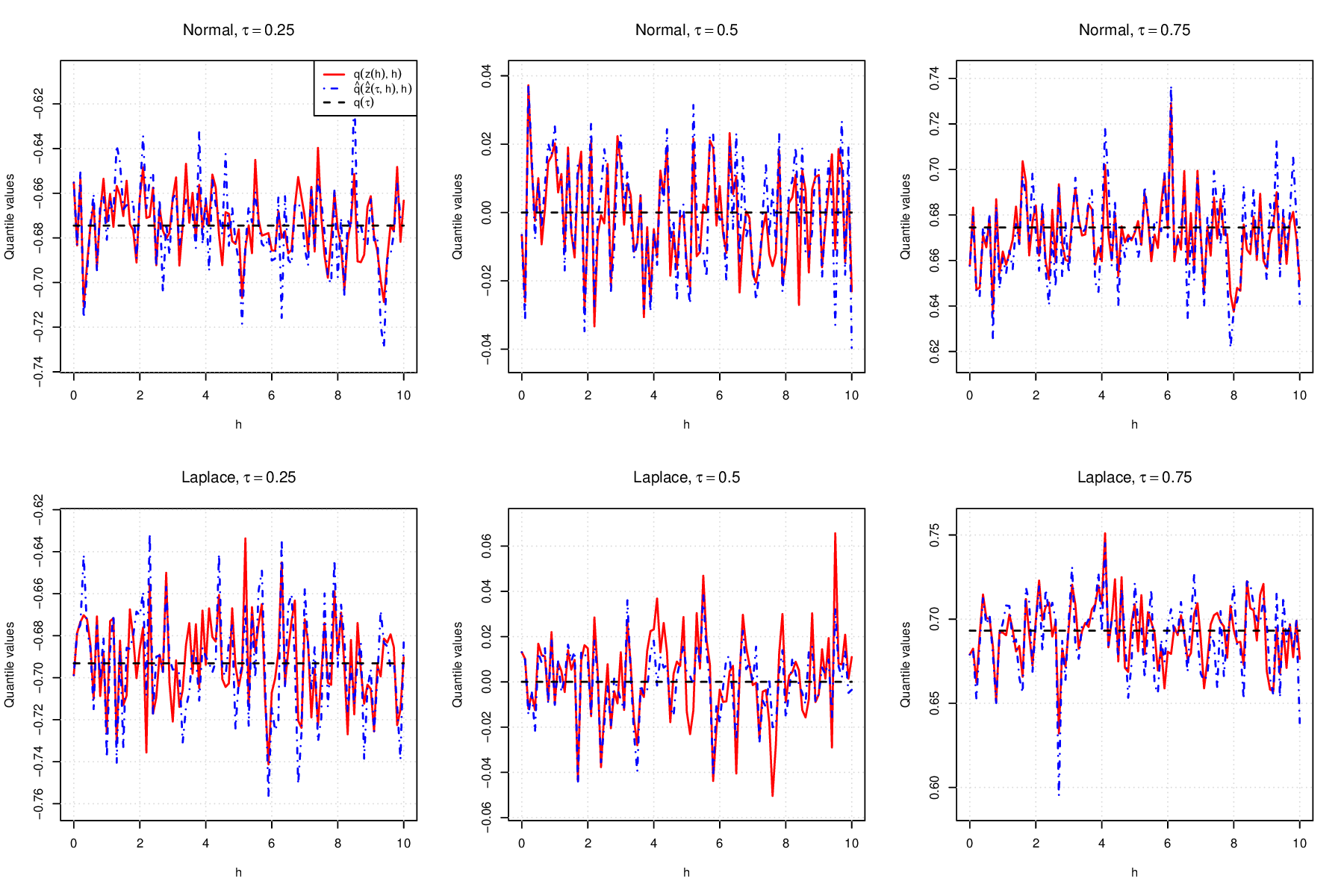}
	\caption{Finite-sample trajectories of \(\hat q(z(\tau,h),h)\) and \(\hat q(\hat z(\tau,h),h)\) for \(\tau=0.25,0.5,0.75\). Top row: Normal distribution. Bottom row: Laplace distribution. The horizontal line indicates the target quantile \(q(\tau)\).}
	\label{fig:fixed_tau_paths}
\end{figure}

\subsection{Experiment 3: Large-\texorpdfstring{$h$}{h} Behaviour and Convergence to the Mean}
\label{subsec:exp_c}

This experiment investigates the large--$h$ behaviour of the estimator family and provides a numerical illustration of the asymptotic equivalence results established in Theorem~\ref{thm:mean-clt-main} and Theorem~\ref{thm:plugin-clt}. In particular, we compare three estimators of the mean $m$:
\begin{enumerate}
\item the sample mean $\bar Y$;
\item the interpolated estimator $\hat q(z,h)$ with a fixed value of $z$;
\item the mean--estimating family $\hat q(\hat z,h)$, where $\hat z$ is constructed from the data.
\end{enumerate}

\paragraph{Design of the experiment.}
Independent samples of size $n=1000$ are generated from two distributions: the standard normal distribution and the standard Laplace distribution. For each Monte Carlo replication and each interpolation level
\[
h \in \{0,1,2,5,10,20,50\},
\]
we compute the three estimators listed above. The experiment is repeated over $500$ Monte Carlo replications, and empirical variances, biases, and mean squared errors are recorded.

\paragraph{Results.}
Figure~\ref{fig:experiment_c} reports the Monte Carlo variances of the three estimators as functions of $h$, for both distributions. Several clear patterns emerge.

First, the variance of the mean--estimating family $\hat q(\hat z,h)$ is essentially constant in $h$ and numerically indistinguishable from the variance of the sample mean $\bar Y$. This is fully consistent with Theorem~\ref{thm:mean-clt-main}, which shows that
\[
\hat q(\hat z,h) = \bar Y + o_p(n^{-1/2})
\quad\text{for every fixed } h\ge 0.
\]

Second, the fixed--$z$ estimator $\hat q(z,h)$ exhibits a markedly different behaviour. Its variance depends strongly on $h$ and converges to the variance of the sample mean as $h$ increases. This convergence is monotone in both distributions considered and reflects the fact that, for fixed $z$, the population solution $q(z,h)$ approaches $m$ as $h\to\infty$.

Third, the contrast between the two estimators highlights an important distinction between population efficiency and implementable efficiency. Although interpolation may improve the population variance along fixed--$\tau$ lines (as discussed in Corollary~\ref{cor:hstar}), the plug--in estimator $\hat q(\hat z,h)$ remains first--order equivalent to the sample mean, and therefore cannot outperform it asymptotically.

\paragraph{Conclusion.}
This experiment confirms the theoretical finding that the mean--estimating family behaves, from a first--order perspective, exactly like the sample mean, uniformly in $h$. At the same time, it illustrates how interpolation affects estimators with fixed tuning parameters, thereby clarifying the distinct roles played by population geometry and plug--in implementation.

\begin{figure}[!htbp]
\centering
\includegraphics[width=\textwidth]{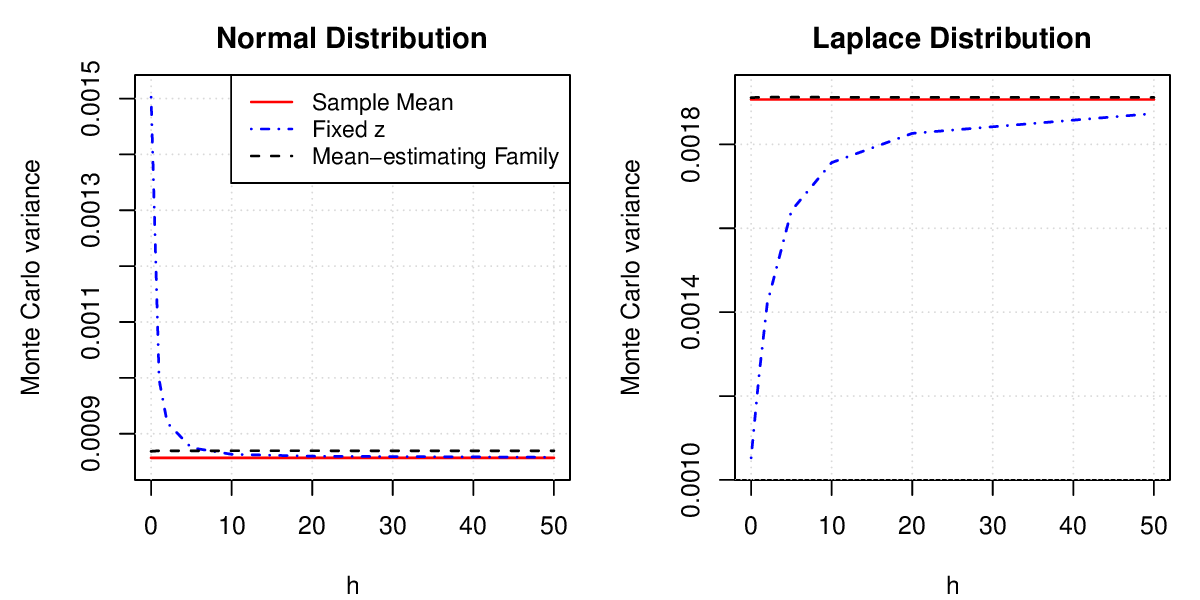}
\caption{Monte Carlo variances of the sample mean, the fixed--$z$ estimator, and the mean--estimating family as functions of $h$, for the normal (left) and Laplace (right) distributions.}
\label{fig:experiment_c}
\end{figure}

\subsection{Empirical Verification of the Monotonicity Properties}
\label{subsec:monotonicity}

\paragraph{Motivation.}
The theoretical analysis of Section~\ref{sec:geometry} shows that, for each fixed $z$,
the population solution $q(z,h)$ moves monotonically toward the mean $m$ as the
interpolation parameter $h$ increases. The direction of this movement depends on the
relative position of the associated quantile $F^{-1}((1-z)/2)$ with respect to $m$.
In particular:
(i) if $F^{-1}((1-z)/2)<m$, then $q(z,h)$ increases in $h$;
(ii) if $F^{-1}((1-z)/2)>m$, then $q(z,h)$ decreases in $h$; and
(iii) if $F^{-1}((1-z)/2)=m$, then $q(z,h)=m$ for all $h\ge0$.

This subsection investigates whether these monotonicity properties persist at the
sample level when the population quantities are replaced by their empirical counterparts.

\paragraph{Design of the experiment.}
To this end, we consider the empirical estimator $\hat q(z,h)$ defined by the estimating
equation~\eqref{eq:empirical-score}, with the unknown distribution function $F$
and mean $m$ replaced by $\hat F$ and $\bar Y$, respectively.
For each dataset, we fix several values of $z$ corresponding to quantile levels
$\tau=(1-z)/2$ below, equal to, and above the mean.
For each such $z$, we compute $\hat q(z,h)$ over a fine grid of $h$ values ranging from
$0$ to a large upper bound.

The experiment is conducted for samples of size $n=5000$ drawn from both the standard
normal and Laplace distributions. The resulting trajectories $h\mapsto\hat q(z,h)$
are compared to the empirical quantile $\hat q(\tau)$ (corresponding to $h=0$) and the
sample mean $\bar Y$ (corresponding to the large-$h$ limit).

\paragraph{Results.}
Figure~\ref{fig:monotonicity} displays the empirical trajectories for the three cases.
In all configurations, the behaviour of $\hat q(z,h)$ closely mirrors the theoretical
properties of the population mapping $q(z,h)$.

When the target quantile lies below the mean, the estimator $\hat q(z,h)$ increases
monotonically from the sample quantile toward the sample mean as $h$ grows.
Conversely, when the target quantile lies above the mean, $\hat q(z,h)$ decreases
monotonically toward $\bar Y$.
Finally, when the target quantile coincides with the mean, the trajectory remains
numerically flat, confirming that $\hat q(z,h)$ is essentially invariant in $h$ in
this case.

These results provide strong empirical evidence that the monotonicity properties
derived at the population level are preserved by the empirical estimator.
They reinforce the interpretation of the family $\{\hat q(z,h):h\ge0\}$ as a continuous
path connecting the sample quantile to the sample mean while retaining the structural
ordering implied by the parameter $z$.

\begin{figure}[!htbp]
\centering
\includegraphics[width=\textwidth]{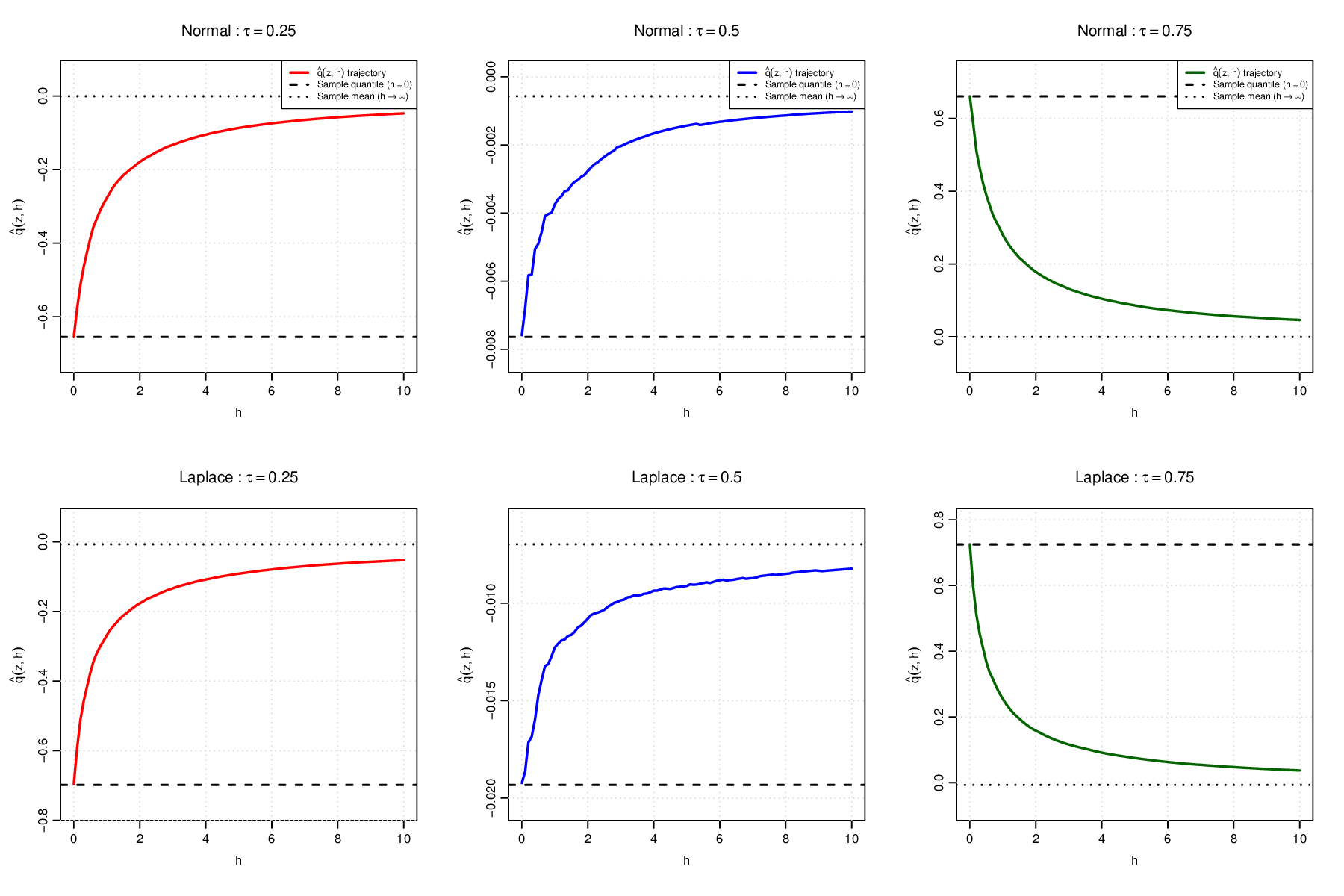}
\caption{Empirical trajectories $h\mapsto \hat q(z,h)$ for quantiles below, equal to,
and above the mean. The dashed line corresponds to the empirical quantile ($h=0$),
while the dotted line indicates the sample mean.}
\label{fig:monotonicity}
\end{figure}

\subsection{Real-data Application}
\label{subsec:realdata}

The preceding experiments established that the empirical estimator
$\hat q_n(z,h)$ preserves the qualitative geometric properties of its
population counterpart $q(z,h)$, and that the interpolation parameter $h$
controls the transition from quantile-like to mean-like behaviour.
We now investigate whether these features persist when the estimator
is applied to real data exhibiting heavy tails and potential departures
from classical parametric assumptions.

For this, we now investigate the empirical behaviour of the Plug--in estimator
\[
\hat q_n(z,h)
\]
on real financial data. We analyse daily closing prices of the CAC–40 index, obtained from Yahoo Finance via the \texttt{quantmod} R package \citep{quantmod}, over the period 2007–2025.  
Let $P_t$ denote the closing price on day $t$. We work with log-returns
\[
Y_t = \log(P_t) - \log(P_{t-1}),
\]
which are standard in financial econometrics. After removing calendar-related gaps and forward and backward filling isolated missing values, the final series contains $n=|\{Y_t\}|$ observations. Figure~\ref{fig:CAC40-diagnostics} provides basic diagnostics of the price and return series.

\begin{figure}[!htbp]
\centering
\includegraphics[width=0.95\textwidth]{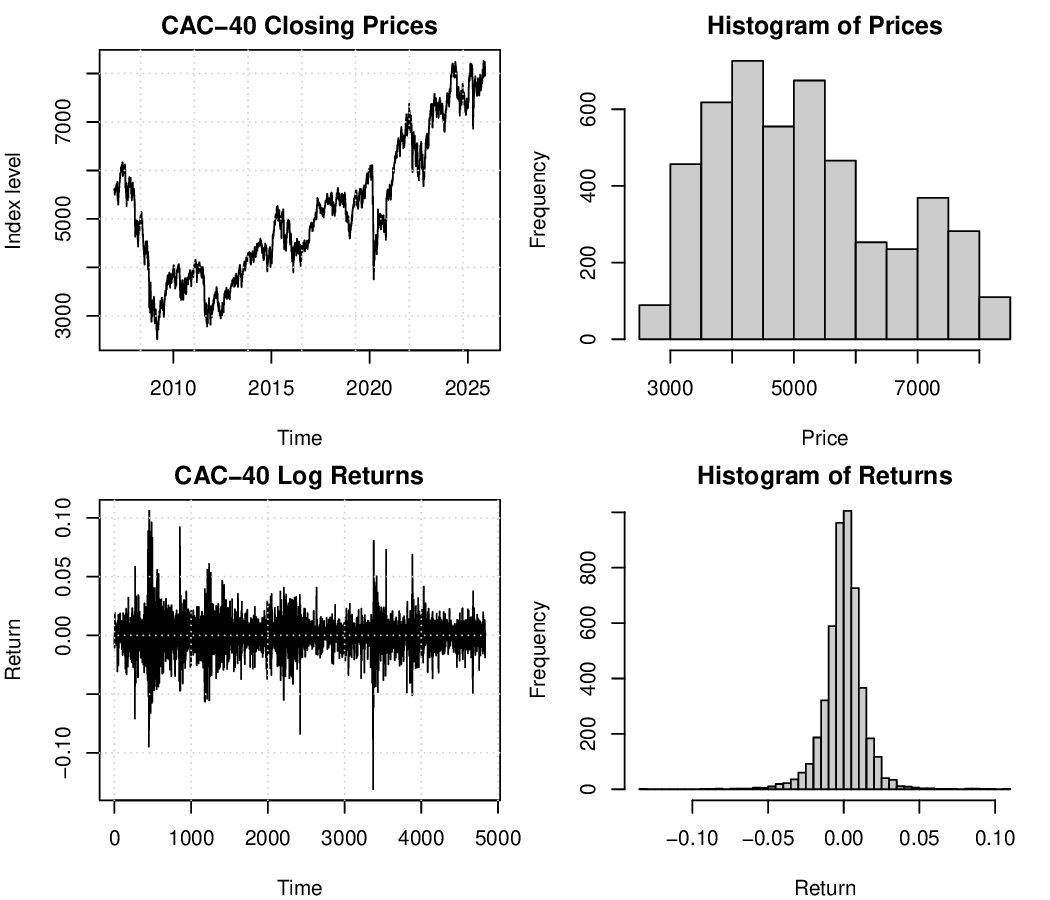}
\caption{Diagnostic plots for CAC–40 daily data (2007–2025). 
Top-left: closing prices; top-right: histogram of prices; 
bottom-left: log-returns; bottom-right: histogram of log-returns.}
\label{fig:CAC40-diagnostics}
\end{figure}

\paragraph{Empirical Plug--in estimator.}
Given the observed returns $Y_1,\dots,Y_n$, the empirical Plug--in estimator is defined by the sample analogue of the population equation:
\begin{equation}
\label{eq:qhat-realdata}
\hat F_n(q) \;+\; \frac{h}{2} q
\;=\;
\frac{1 - z + h\,\bar Y_n}{2},
\end{equation}
where $\hat F_n$ is the empirical cdf and $\bar Y_n$ is the sample mean.  
For each $(z,h)$, we solve~\eqref{eq:qhat-realdata} numerically by root-finding,
using a conservative bracketing interval 
\[
q \in \bigl[\min(Y)-10\,\widehat\sigma,\; \max(Y)+10\,\widehat\sigma\bigr],
\]
with $\widehat\sigma$ the sample standard deviation.

Because financial returns have heavier tails than Gaussian or Laplace samples, the function
$h \mapsto \hat q_n(z,h)$ stabilises more slowly as $h$ increases.  
In particular, for the CAC–40 dataset considered here, the curves do not fully reach their limiting affine shape within $h\in[0,5]$,
unlike in the simulation study.
To make the stabilisation clearly visible, we examine values up to
\[
h \in [0,200],\qquad h = 0,0.1,0.2,\dots,200.
\]

\paragraph{Empirical results.}

For three representative values $z\in\{-0.5,0,0.5\}$, we compute the full mapping
$h\mapsto\hat q_n(z,h)$.
The resulting trajectories are shown in Figure~\ref{fig:CAC40-qhat}.

\begin{figure}[!htbp]
\centering
\includegraphics[width=\textwidth]{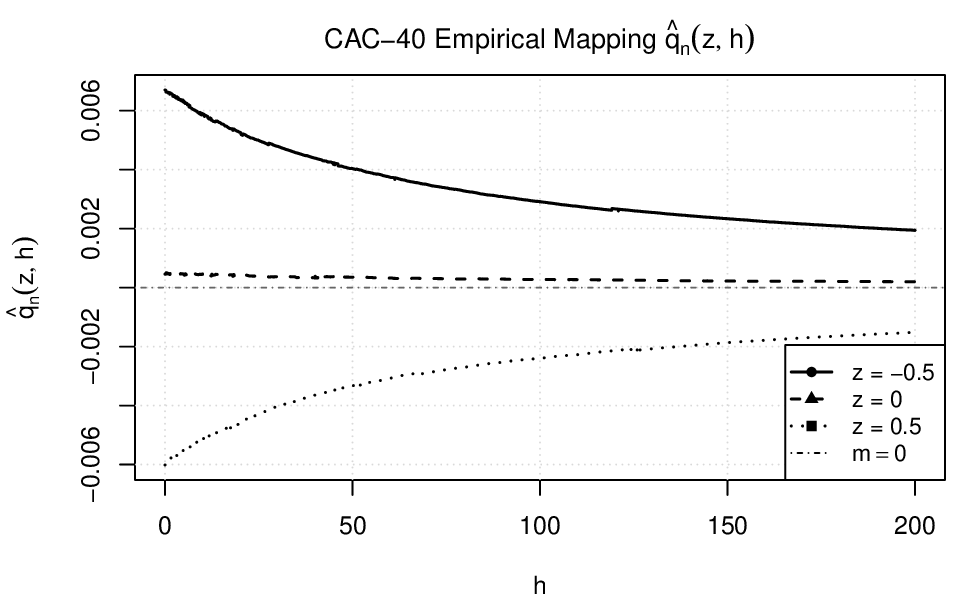}
\caption{Empirical behaviour of $\hat q_n(z,h)$ computed from CAC–40 daily returns.
Three curves shown: $z=-0.5$ (solid), $z=0$ (dashed), and $z=0.5$ (dotted). 
Horizontal line marks the sample mean $\bar Y_n$.}
\label{fig:CAC40-qhat}
\end{figure}

Several features mirror those observed in the population and simulation analyses:

\begin{itemize}
\item For $z=-0.5$, the estimator $\hat q_n(z,h)$ is positive and decreases monotonically in $h$.
\item For $z=0.5$, the estimator is negative and increases with $h$.
\item The qualitative behaviour is consistent with the theoretical mapping $q(z,h)$:
the $z<0$ and $z>0$ curves have opposite monotonic trends in $h$, while the neutral curve ($z=0$) remains close to zero.
\end{itemize}

A key difference with the simulation study is the rate at which stabilisation occurs.  
For Gaussian or Laplace samples, the mapping $h\mapsto\hat q_n(z,h)$ typically reaches its limiting regime around $h\approx 4$.  
For real financial data, heavier-tailed behaviour makes the stabilisation substantially slower: the curves settle only for $h$ in the order of 50--100.  
This behaviour is not a defect of the estimator but a direct reflection of the empirical distribution of returns, whose tails decay more slowly than in classical parametric models.

Overall, the real-data experiment confirms that:
\begin{enumerate}
\item The empirical Plug--in estimator $\hat q_n(z,h)$ behaves in a manner consistent with the theoretical function $q(z,h)$.
\item The interpolation parameter $h$ must be chosen in relation to the tail behaviour of the underlying distribution: heavier-tailed data require larger~$h$ before the linear asymptotics fully appear.
\end{enumerate}
These findings reinforce the relevance of the Plug--in formulation and clarify how interpolation interacts with distributional features in practical applications.

The previous experiments focused on the path \(h \mapsto \hat q_n(z,h)\) itself.
We now complement these findings by examining how interpolation influences the estimator’s asymptotic variance.

This analysis is key for understanding the estimator’s efficiency, and it prepares the ground for the next subsection, where the optimal choice of \(h\) is studied systematically across quantiles.
\section{Conclusion}
\label{sec:conclusion}
This paper develops a unified framework for understanding and comparing three fundamental classes of estimators:
\begin{enumerate}[label=(\roman*)]
	\item the \emph{fixed-parameter interpolated quantile estimators} $\hat q(z,h)$;
    \item the \emph{fixed-quantile Plug--in estimator}; and
    \item the \emph{mean-estimating family} obtained by evaluating the interpolated map at 
          $\hat z = 1-2\hat F(\bar Y)+h\bigl(\bar Y- \hat q(\hat F(\bar Y))\bigr)$.
\end{enumerate}
All three arise from the same interpolated objective function and admit a common M-estimation structure.

A central contribution of the paper is the derivation of the \emph{three corresponding central limit theorems}, each established through a unified proof strategy relying on \emph{uniform asymptotic equicontinuity}.
This provides a coherent theoretical foundation in which the three estimator classes differ only through the way the pair $(z,h)$ is chosen—fixed, data-driven, or lying along the geometric lines associated with a fixed quantile level.

The parameter geometry plays a key organizational role.
For each quantile level $\tau$, the admissible parameter pairs $(z,h)$ lie on straight lines along which the population target remains $q(\tau)$ while the asymptotic variance evolves according to the explicit formula $v(\tau,h)$.
This geometric structure clarifies how interpolation affects efficiency without altering the target of estimation, and it naturally separates the behaviour of the three estimator classes.

The asymptotic efficiency analysis reveals two regimes.
At the population level, the classical interpolated estimator may exhibit a finite optimal interpolation level $h^*(\tau)$ for heavy-tailed distributions (case~(c) of Corollary~\ref{cor:hstar}), whereas in light-tailed settings interpolation reduces variance monotonically.
In contrast, the Plug--in estimator $\hat q(\hat z(\tau,h),h)$ behaves differently: its asymptotic variance is \emph{always} strictly decreasing in $h$, with no finite minimiser, and converges to $\Var(Y)$.
This explains why the population and Plug--in variances must be treated separately and why Corollary~\ref{cor:hstar} remains essential despite the Plug--in correction eliminating interior optima.

The numerical section confirms all theoretical predictions.
For the fixed-quantile Plug--in estimator, both simulated and real data display empirical trajectories $h\mapsto \hat q_n(z,h)$ that match the theoretical shape of $q(z,h)$, with stabilization occurring more slowly under heavy tails.
For the interpolated quantile estimator, the efficiency patterns illustrated numerically coincide exactly with the three cases of Corollary~\ref{cor:hstar}.
Finally, for the mean-estimating family $\hat q(\hat z,h)$, both the theory and a dedicated Monte-Carlo study confirm that \emph{the estimator is first-order equivalent to the sample mean for every $h\ge0$}, showing no variance reduction from interpolation.

Overall, the paper provides a complete and unified treatment of three estimator classes that have traditionally been studied separately.
The framework reveals a coherent geometric and asymptotic structure governing robustness, interpolation, and efficiency, and offers a principled basis for further developments.
Future work will extend the analysis to interpolated quantile regression, optimal data-driven selection of $h$, and adaptive estimation under heavy-tailed or asymmetric noise.

\paragraph{Data Availability Statement} 
The real-world financial data analysed in Section~\ref{subsec:realdata} are publicly available from Yahoo Finance (CAC-40 index). The simulation code used to 
produce the results, along with code for downloading and processing the 
financial data, is available from the corresponding author upon reasonable request.

\paragraph{Funding Statement} This research received no specific grant from any funding agency in the public, commercial, or not-for-profit sectors.

\paragraph{Conflict of Interest Statement} The authors declare that there are no conflicts of interest regarding the publication of this paper.

\clearpage
\bibliographystyle{plainnat}
\bibliography{references}

\begin{thebibliography}{17}
\providecommand{\natexlab}[1]{#1}
\providecommand{\url}[1]{\texttt{#1}}
\expandafter\ifx\csname urlstyle\endcsname\relax
  \providecommand{\doi}[1]{doi: #1}\else
  \providecommand{\doi}{doi: \begingroup \urlstyle{rm}\Url}\fi

\bibitem[Bellini et~al.(2014)Bellini, Klar, Müller, and
  Rosazza~Gianin]{Bellini2014}
Fabio Bellini, Bernhard Klar, Alfred Müller, and Emanuela Rosazza~Gianin.
\newblock Generalized quantiles as risk measures.
\newblock \emph{Insurance: Mathematics and Economics}, 54:\penalty0 41--48,
  January 2014.
\newblock ISSN 0167-6687.
\newblock \doi{10.1016/j.insmatheco.2013.10.015}.

\bibitem[Daouia et~al.(2018)Daouia, Gijbels, and Stupfler]{Daouia2018}
Abdelaati Daouia, Irène Gijbels, and Gilles Stupfler.
\newblock Extremiles: A new perspective on asymmetric least squares.
\newblock \emph{Journal of the American Statistical Association}, 114\penalty0
  (527):\penalty0 1366--1381, October 2018.
\newblock ISSN 1537-274X.
\newblock \doi{10.1080/01621459.2018.1498348}.

\bibitem[Dermoune et~al.(2017)Dermoune, Ounaissi, and Rahmania]{Dermoune2017}
Azzouz Dermoune, Daoud Ounaissi, and Nadji Rahmania.
\newblock Oscillation of metropolis–hastings and simulated annealing
  algorithms around lasso estimator.
\newblock \emph{Mathematics and Computers in Simulation}, 135:\penalty0 39--50,
  2017.
\newblock ISSN 0378-4754.
\newblock \doi{https://doi.org/10.1016/j.matcom.2015.09.003}.
\newblock URL
  \url{https://www.sciencedirect.com/science/article/pii/S0378475415001901}.
\newblock Special Issue: 9th IMACS Seminar on Monte Carlo Methods.

\bibitem[Hampel et~al.(2005)Hampel, Ronchetti, Rousseeuw, and
  Stahel]{Hampel2005}
Frank~R. Hampel, Elvezio~M. Ronchetti, Peter~J. Rousseeuw, and Werner~A.
  Stahel.
\newblock \emph{Robust Statistics: The Approach Based on Influence Functions}.
\newblock Wiley, March 2005.
\newblock ISBN 9781118186435.
\newblock \doi{10.1002/9781118186435}.

\bibitem[Hjort and Pollard(2011)]{Hjort2011}
Nils~Lid Hjort and David Pollard.
\newblock Asymptotics for minimisers of convex processes, 2011.

\bibitem[Huber(1981)]{Huber1981}
Peter~J. Huber.
\newblock \emph{Robust statistics}.
\newblock Wiley series in probability and mathematical statistics. Wiley, New
  York, 1981.
\newblock ISBN 9780471725244.
\newblock Includes bibliographical references and index.

\bibitem[Huber and Ronchetti(2009)]{Huber2009}
Peter~J. Huber and Elvezio~M. Ronchetti.
\newblock \emph{Robust Statistics}.
\newblock Wiley, January 2009.
\newblock ISBN 9780470434697.
\newblock \doi{10.1002/9780470434697}.

\bibitem[Knight(1998)]{Knight1998}
Keith Knight.
\newblock Limiting distributions for $l\sb 1$ regression estimators under
  general conditions.
\newblock \emph{The Annals of Statistics}, 26\penalty0 (2), April 1998.
\newblock ISSN 0090-5364.
\newblock \doi{10.1214/aos/1028144858}.

\bibitem[Koenker and Bassett(1978)]{Koenker1978}
Roger Koenker and Gilbert Bassett.
\newblock Regression quantiles.
\newblock \emph{Econometrica}, 46\penalty0 (1):\penalty0 33, January 1978.
\newblock ISSN 0012-9682.
\newblock \doi{10.2307/1913643}.

\bibitem[Newey and Powell(1987)]{Newey1987}
Whitney~K. Newey and James~L. Powell.
\newblock Asymmetric least squares estimation and testing.
\newblock \emph{Econometrica}, 55\penalty0 (4):\penalty0 819, July 1987.
\newblock ISSN 0012-9682.
\newblock \doi{10.2307/1911031}.

\bibitem[Pollard(1991)]{Pollard1991}
David Pollard.
\newblock Asymptotics for least absolute deviation regression estimators.
\newblock \emph{Econometric Theory}, 7\penalty0 (2):\penalty0 186--199, June
  1991.
\newblock ISSN 1469-4360.
\newblock \doi{10.1017/s0266466600004394}.

\bibitem[Ryan and Ulrich(2025)]{quantmod}
Jeffrey~A. Ryan and Joshua~M. Ulrich.
\newblock \emph{quantmod: Quantitative Financial Modelling Framework}, 2025.
\newblock URL \url{https://CRAN.R-project.org/package=quantmod}.
\newblock R package version 0.4.28.

\bibitem[Tibshirani(1996)]{Tibshirani1996}
Robert Tibshirani.
\newblock Regression shrinkage and selection via the lasso.
\newblock \emph{Journal of the Royal Statistical Society Series B: Statistical
  Methodology}, 58\penalty0 (1):\penalty0 267--288, January 1996.
\newblock ISSN 1467-9868.
\newblock \doi{10.1111/j.2517-6161.1996.tb02080.x}.

\bibitem[van~der Vaart(1998)]{Vaart1998}
Aad~W. van~der Vaart.
\newblock \emph{Asymptotic Statistics}.
\newblock Cambridge University Press, October 1998.
\newblock ISBN 9780521784504.
\newblock \doi{10.1017/cbo9780511802256}.

\bibitem[van~der Vaart and Wellner(1996)]{Vaart1996}
Aad~W. van~der Vaart and Jon~A. Wellner.
\newblock \emph{Weak Convergence and Empirical Processes}.
\newblock Springer New York, 1996.
\newblock ISBN 9781475725452.
\newblock \doi{10.1007/978-1-4757-2545-2}.

\bibitem[Waltrup et~al.(2014)Waltrup, Sobotka, Kneib, and
  Kauermann]{Waltrup2014}
Linda~Schulze Waltrup, Fabian Sobotka, Thomas Kneib, and Göran Kauermann.
\newblock Expectile and quantile regression—david and goliath?
\newblock \emph{Statistical Modelling}, 15\penalty0 (5):\penalty0 433--456,
  December 2014.
\newblock ISSN 1477-0342.
\newblock \doi{10.1177/1471082x14561155}.

\bibitem[Zou(2006)]{Zou2006}
Hui Zou.
\newblock The adaptive lasso and its oracle properties.
\newblock \emph{Journal of the American Statistical Association}, 101\penalty0
  (476):\penalty0 1418--1429, December 2006.
\newblock ISSN 1537-274X.
\newblock \doi{10.1198/016214506000000735}.

\end{thebibliography}
\clearpage
\appendix
\section{Proof of Theorem~\ref{thm:CLT}: Interpolated Quantile CLT}
\label{app:clt1}
This appendix provides a complete proof of the asymptotic normality of the interpolated quantile estimators
\[
\hat q(z,h)=\arg\min_{q\in\mathbb{R}} \hat M(q;z,h),
\qquad
z\in(-1,1),h\ge 0,
\]
where
\[
\hat M(q;z,h)
=\frac{1}{n}\sum_{i=1}^n m(q-Y_i;z,h),
\qquad
m(u;z,h)=|u|+z u+\tfrac{h}{2}u^2.
\]
The proof follows the structure of a uniform Z–estimator argument, where the key ingredients are:
\begin{itemize}
    \item Knight’s identity for the absolute value function;
    \item a uniform control of the empirical process via asymptotic equicontinuity;
    \item a quadratic expansion of the criterion;
    \item uniqueness and stability of the minimiser.
\end{itemize}
Throughout we impose assumptions \(A1\)–\(A2\) from Section~\ref{sec:asymptotics}.
\subsection*{Preliminaries}
Let the population criterion be
\[
M(q;z,h)=\mathbb{E}[m(q-Y;z,h)],
\]
and let \(q(z,h)\) be its unique minimiser, equivalently the unique solution of the population score equation
\[
\Psi(q;z,h)=0,
\qquad
\Psi(q;z,h)
=\mathbb{E}\bigl[\operatorname{sgn}(q-Y)\bigr] + z + h(m-q).
\]
The empirical score is
\[
\hat\Psi(q;z,h) = 2\hat F(q) - 1 + z + h(q - \bar Y).
\]
This is equivalent to $\frac{1}{n}\sum_{i=1}^n \operatorname{sgn}(q-Y_i) + z + h(q - \bar Y)$.
Because \(f\) is continuous and strictly positive around \(q(z,h)\), the mapping
\[
q\mapsto \Psi(q;z,h)
\]
is strictly increasing, with derivative \(2f(q(z,h))+h\).
\subsection*{Knight’s Identity and Empirical Expansion}
Knight’s identity states:
\[
|y-(q+\delta)|-|y-q|
= -\delta,\sgn(y-q)
+ 2\int_0^\delta
  \bigl(1_{{y\le q+s}}-1_{{y\le q}}\bigr)ds.
\]
Plugging this into \(\hat M\bigl(q(z,h)+\delta;z,h\bigr)-\hat M\bigl(q(z,h);z,h\bigr)\) gives
\begin{align}
\label{eq:A.basic}
&\hat M(q(z,h)+\delta;z,h)-\hat M(q(z,h);z,h)\\
&=\delta\hat\Psi(q(z,h);z,h)
+ 2\int_0^\delta \bigl(\widehat F(q(z,h)+s)-\widehat F(q(z,h))\bigr)ds
+\frac{h}{2}\delta^2.
\nonumber
\end{align}
We now analyse this display under the local scaling \(\delta=t/\sqrt{n}\).
\subsection*{Uniform Asymptotic Equicontinuity}
Let
\[
\mathbb{G}_n = \sqrt{n}(\widehat F - F)
\]
denote the empirical process.

Because \(F\) admits a continuous, strictly positive density \(f\) around \(q(z,h)\), the expansion
\[
F(q+s)-F(q)=f(q)s + o(s)
\]
holds uniformly over small \(s\).

Since the empirical process is uniformly asymptotically equicontinuous on compact intervals (Donsker property of \({1_{{y\le x}}:x\in\mathbb{R}})\), we obtain the uniform approximation:
\[
\sup_{|t|\le T}\left|
\int_0^{t/\sqrt{n}}
(\widehat F(q+s)-\widehat F(q)),ds
-
\frac{f(q)}{2}\frac{t^2}{n}
\right|
\xrightarrow{p}0,
\tag{A.1}\label{eq:A.equicontinuity}
\]
uniformly over \((z,h)\) in compact subsets of \((-1,1)\times[0,\infty)\).

This is the key uniformity enabling Z–estimator linearisation.

\subsection*{Local Expansion of the Empirical Criterion}
Using \eqref{eq:A.basic} with \(\delta=t/\sqrt{n}\) and the equicontinuity result \eqref{eq:A.equicontinuity}, we obtain uniformly over bounded \(t\):
\begin{align}
&\hat M\left(q(z,h)+\frac{t}{\sqrt{n}};z,h\right)
-\hat M(q(z,h);z,h)\\
&=\frac{t}{\sqrt{n}}\hat\Psi(q(z,h);z,h)
+\Bigl( f(q(z,h))+\tfrac{h}{2}\Bigr)\frac{t^2}{n}
+ r_n(z,h,t),
\nonumber
\end{align}
where
\[
\sup_{|t|\le T}|r_n(z,h,t)|\xrightarrow{p}0.
\tag{A.2}
\]
Thus the empirical criterion admits a uniform quadratic expansion around its minimiser.
\subsection*{Local Minimisation and Linear Representation}
We observe that the minimizer of the function
\[
t \mapsto \hat M\!\left(q(z,h)+\frac{t}{\sqrt{n}};z,h\right) - \hat M(q(z,h);z,h)
\]
is exactly \(\sqrt{n}\bigl(\hat q(z,h)-q(z,h)\bigr)\).
From the quadratic expansion, this minimizer is close to the minimizer of the quadratic form
\[
t \mapsto \frac{t}{\sqrt{n}}\,\hat\Psi(q(z,h);z,h)
+ \Bigl( f(q(z,h)) + \frac{h}{2} \Bigr)\frac{t^2}{n}.
\]
Solving this quadratic minimization yields
\[
\sqrt{n}\bigl(\hat q(z,h)-q(z,h)\bigr) = -\frac{\sqrt{n}\,\hat\Psi(q(z,h);z,h)}{2f(q(z,h))+h} + o_p(1).
\tag{A.3}\label{eq:A.linear}
\]
Equivalently, since \(\partial_z q(z,h) = -1/(2f(q(z,h))+h)\), we have
\[
\sqrt{n}\bigl(\hat q(z,h)-q(z,h)\bigr) = \partial_z q(z,h) \sqrt{n}\,\hat\Psi(q(z,h);z,h) + o_p(1).
\]
Now it suffices to establish a central limit theorem for \(\sqrt{n}\,\hat\Psi(q(z,h);z,h)\)
(see, e.g.,~\cite{Pollard1991,Hjort2011,Knight1998}).
\subsection*{Asymptotic Normality}
Define the influence function:
\[
\psi(q(z,h)-Y_i;z,h) = \operatorname{sgn}(q(z,h)-Y_i) + h(q(z,h)-Y_i) + z.
\]
Because \(Y\) has a finite second moment and \(f\) is continuous near \(q(z,h)\), this function has finite variance and mean zero.

From the definition of \(\hat\Psi\):
\[
\sqrt{n},\hat\Psi(q(z,h);z,h)
= \frac{1}{\sqrt{n}}\sum_{i=1}^n
\psi(Y_i;q(z,h),z,h).
\]
By the classical Lindeberg–Feller theorem,
\[
\frac{1}{\sqrt{n}}\sum_{i=1}^n
\psi(Y_i;q(z,h),z,h)
\xrightarrow{d}
\mathcal{N}(0,B(z,h)),
\]
where
\[
B(z,h)=\Var!\bigl(\psi(Y;q(z,h),z,h)\bigr).
\]
Combining this with the linear representation \eqref{eq:A.linear} yields
\[
\sqrt{n}\bigl(\hat q(z,h)-q(z,h)\bigr)
\xrightarrow{d}
\mathcal{N}\left(
0,
\frac{B(z,h)}{(2f(q(z,h))+h)^2}
\right).
\]
This completes the proof of Theorem~\ref{thm:CLT}.
\hfill
$\square$
\subsection*{Closed-form Expression for \(B(z,h)\)}
By a direct calculation using independence and the identity
\(\mathbb{E}[\sgn(Y-q)]=1-2F(q)\),
\begin{align*}
&B(z,h)
=4F(q(z,h))(1-F(q(z,h))) \\
& +2h\Bigl[\mathbb{E}|Y-q(z,h)|
-(m-q(z,h))(1-2F(q(z,h)))\Bigr]
+h^2\Var(Y),
\end{align*}
as stated in the theorem.
\section{Proof of Theorem~\ref{thm:plugin-clt}: Plug--in Quantile CLT}
\label{app:plugin-clt}
This appendix derives the asymptotic expansion and central limit theorem for the
implementable estimator
\[
\hat q\bigl(\hat z(\tau,h),h\bigr),
\qquad
\hat z(\tau,h)=1-2\tau+h\bigl(\bar Y-\hat q(\tau)\bigr),
\]
with population target
\[
q(z(\tau,h),h)=q\bigl(z(\tau,h),h\bigr),
\qquad
z(\tau,h)=1-2\tau+h\bigl(m-F^{-1}(\tau)\bigr).
\]
The goal is to obtain a linear representation of
\[
\sqrt{n}\bigl(\hat q(\hat z(\tau,h),h)-q(z(\tau,h),h)\bigr)
\]
and identify the corresponding influence function and asymptotic variance.
\subsection*{Linearisation for the Plug--in Estimator}
Applying the linear representation from Appendix~\ref{app:clt1} (Theorem~\ref{thm:CLT}) to the random parameter \(\hat z(\tau,h)\) gives directly
\begin{equation}
\begin{aligned}
\sqrt n\bigl(\hat q(\hat z(\tau,h),h)-q(\hat z(\tau,h),h)\bigr)
&=
-\frac{1}{2f(q(\hat z(\tau,h),h))+h}\,
  \sqrt n\,\hat\Psi\!\bigl(q(\hat z(\tau,h),h);\hat z(\tau,h),h\bigr)
\\
&\quad
+\,o_p(1).
\end{aligned}
\tag{B.1}\label{eq:plug-start}
\end{equation}

\paragraph{Expansion of \(q(\hat z(\tau,h),h)\).}
By the differentiability of \(q(z,h)\) with respect to \(z\),
\begin{equation}
q(\hat z(\tau,h),h)
=
q(z(\tau,h),h) + \frac{\partial q}{\partial z}(z(\tau,h),h)\bigl(\hat z(\tau,h)-z(\tau,h)\bigr) + o_p(n^{-1/2}),
\tag{B.2}\label{eq:q-expansion}
\end{equation}
where differentiating \(\Psi(q(z,h);z,h)=0\) gives
\[
\frac{\partial q}{\partial z}(z(\tau,h),h)
=
-\frac{1}{2f(q(z(\tau,h),h))+h}.
\]
\paragraph{Increment expansion of the empirical score.}
\begin{lemma}[Increment expansion of the empirical score]
\label{lem:increment}
Under Assumptions (A1)--(A2), for any sequences \(q_1 = q(z(\tau,h),h) + O_p(n^{-1/2})\), 
\(q_2 = q(z(\tau,h),h) + O_p(n^{-1/2})\), and \(\hat z_n = z^* + O_p(n^{-1/2})\), we have
\[
\hat\Psi(q_1; \hat z, h) - \hat\Psi(q_2; \hat z, h)
= (2f(q(z(\tau,h),h)) + h)(q_1 - q_2) + o_p(n^{-1/2}).
\]
\end{lemma}
\begin{proof}
Using the canonical representation \(\hat\Psi(q; z, h) = 2\hat F(q) - 1 + z + h(q - \bar Y)\),
\[
\hat\Psi(q_1; \hat z, h) - \hat\Psi(q_2; \hat z, h)
= 2\bigl[\hat F(q_1) - \hat F(q_2)\bigr] + h(q_1 - q_2).
\tag{E.1}
\]
Decompose the empirical CDF increment:
\[
\hat F(q_1) - \hat F(q_2)
= \bigl[\hat F(q_1) - F(q_1)\bigr]
+ \bigl[F(q_1) - F(q_2)\bigr]
+ \bigl[F(q_2) - \hat F(q_2)\bigr].
\]
We analyse the three terms separately:
\begin{enumerate}
    \item \textbf{Differentiability of \(F\)}: By the mean value theorem, there exists \(\tilde{q}\) between \(q_1\) and \(q_2\) such that
    \[
    F(q_1) - F(q_2) = f(\tilde{q})(q_1 - q_2).
    \]
    Since \(q_1, q_2 = q(z(\tau,h),h) + O_p(n^{-1/2})\), we have \(\tilde{q} \xrightarrow{p} q(z(\tau,h),h)\). By continuity of \(f\) at \(q(z(\tau,h),h)\) (Assumption A2), \(f(\tilde{q}) = f(q(z(\tau,h),h)) + o_p(1)\). Thus,
    \[
    F(q_1) - F(q_2) = f(q(z(\tau,h),h))(q_1 - q_2) + o_p(1)(q_1 - q_2).
    \]
    Because \(q_1 - q_2 = O_p(n^{-1/2})\), the term \(o_p(1)(q_1 - q_2)\) is \(o_p(n^{-1/2})\). Hence,
    \[
    F(q_1) - F(q_2)= f(q(z(\tau,h),h))(q_1 - q_2) + o_p(n^{-1/2}).
    \]
    \item \textbf{Asymptotic equicontinuity of the empirical process}: 
    The class \(\mathcal{F} = \{\mathbf{1}_{\{y \le x\}} : x \in \mathbb{R}\}\) is Donsker 
    \citep[Theorem~2.5.2]{Vaart1996}. Consequently, the empirical process 
    \(\mathbb{G}_n = \sqrt{n}(\hat F - F)\) is uniformly asymptotically equicontinuous: 
    for every \(\eta > 0\),
    \[
    \lim_{\delta \downarrow 0} \limsup_{n\to\infty} 
    P\!\left( \sup_{\substack{x_1,x_2\in\mathbb{R}\\ |x_1-x_2|\le\delta}} 
    |\mathbb{G}_n(x_1) - \mathbb{G}_n(x_2)| > \eta \right) = 0.
    \]
    Since \(|q_1 - q_2| = O_p(n^{-1/2})\), there exists a constant \(M > 0\) such that 
    \(|q_1 - q_2| \le M n^{-1/2}\) with probability approaching 1. Taking \(\delta = M n^{-1/2}\) 
    in the equicontinuity statement yields
    \[
    \mathbb{G}_n(q_1) - \mathbb{G}_n(q_2) = o_p(1),
    \]
    i.e.,
    \[
    \sqrt{n}\bigl[(\hat F(q_1) - F(q_1)) - (\hat F(q_2) - F(q_2))\bigr] = o_p(1).
    \]
    Dividing by \(\sqrt{n}\) gives the required rate:
    \[
    (\hat F(q_1) - F(q_1)) - (\hat F(q_2) - F(q_2)) = o_p(n^{-1/2}).
    \]
\end{enumerate}
Combining these results,
\[
\hat F(q_1) - \hat F(q_2) 
= f(q^*)(q_1 - q_2) + o_p(n^{-1/2}).
\]
Substituting into (E.1) yields
\begin{align*}
\hat\Psi(q_1; \hat z, h) - \hat\Psi(q_2; \hat z, h)
&= 2f\bigl(q(z(\tau,h),h)\bigr)(q_1 - q_2) + h(q_1 - q_2) + o_p(n^{-1/2}) \nonumber\\
&= (2f\bigl(q(z(\tau,h),h)\bigr) + h)(q_1 - q_2) + o_p(n^{-1/2}).
\end{align*}
which completes the proof.
\end{proof}

\noindent
Applying Lemma~\ref{lem:increment} with \(q_1 = q(\hat z(\tau,h), h)\) and \(q_2 = q(z(\tau,h),h)\) gives
\begin{align}
\hat\Psi(q(\hat z(\tau,h), h); \hat z(\tau,h), h)
&= \hat\Psi(q(z(\tau,h),h); \hat z(\tau,h), h) \nonumber\\
&\quad + (2f(q(z(\tau,h),h)) + h)\bigl(q(\hat z(\tau,h), h) - q(z(\tau,h),h)\bigr) \nonumber\\
&\quad + o_p(n^{-1/2}).
\tag{B.3}\label{eq:plug-step2}
\end{align}
From the expansion in \eqref{eq:q-expansion}, we have
\[
q(\hat z(\tau,h),h)-q(z(\tau,h),h)
=
-\frac{1}{2f(q(z(\tau,h),h))+h}\bigl(\hat z(\tau,h)-z(\tau,h)\bigr)
+ o_p(n^{-1/2}).
\]
Substituting this into \eqref{eq:plug-step2} yields
\[
\hat\Psi(q(\hat z(\tau,h),h);\hat z(\tau,h),h)
=
\hat\Psi(q(z(\tau,h),h);\hat z(\tau,h),h)
-
\bigl(\hat z(\tau,h)-z(\tau,h)\bigr)
+ o_p(n^{-1/2}).
\]
Now, using the linearity of \(\hat\Psi\) in its second argument,
\[
\hat\Psi(q(z(\tau,h),h);\hat z(\tau,h),h)
= \hat\Psi(q(z(\tau,h),h); z(\tau,h),h) 
+ \bigl(\hat z(\tau,h)-z(\tau,h)\bigr)
+ o_p(n^{-1/2}).
\]
Combining these two expressions, we obtain the key simplification:
\[
\hat\Psi(q(\hat z(\tau,h),h);\hat z(\tau,h),h)
= \hat\Psi(q(z(\tau,h),h); z(\tau,h),h) 
+ o_p(n^{-1/2}).
\]
Substituting this identity into \eqref{eq:plug-start} gives immediately
\begin{align*}
\sqrt{n}\Bigl(\hat q(\hat z(\tau,h),h)-q(z(\tau,h),h)\Bigr)
=&-\frac{1}{2f(q(z(\tau,h),h))+h}\sqrt{n}\,\hat\Psi(q(z(\tau,h),h); z(\tau,h),h) \\
&-\frac{1}{2f(q(z(\tau,h),h))+h}\sqrt{n}\bigl(\hat z(\tau,h)-z(\tau,h)\bigr) \\
&+o_p(1),
\tag{B.4}\label{eq:q-final}
\end{align*}

\medskip
\noindent
Equation~\eqref{eq:q-final} reveals the structure of the Plug--in estimator's asymptotic distribution.
The first term,
\[
-\frac{1}{2f(q(z(\tau,h),h))+h}\,
\sqrt{n}\,\hat\Psi(q(z(\tau,h),h); z(\tau,h),h),
\]
corresponds exactly to the expansion we would obtain from Theorem~\ref{thm:CLT} 
if we treated $z$ as fixed and equal to $z(\tau,h)$.  
The second term,
\[
-\frac{1}{2f(q(z(\tau,h),h))+h}\,
\sqrt{n}\bigl(\hat z(\tau,h)-z(\tau,h)\bigr),
\]
is the correction that accounts for estimating the parameter $z(\tau,h)$ by 
$\hat z(\tau,h)$.  
Thus, Theorem~\ref{thm:plugin-clt} can be viewed as Theorem~\ref{thm:CLT} applied at 
$z=z(\tau,h)$, plus an additional contribution arising from the Plug--in step.

\medskip
\noindent
This completes the derivation of the linear representation for 
$\hat q(\hat z(\tau,h),h)$.
\subsection*{Expansion of the Plug--in Parameter $\hat z(\tau,h)$}
By definition,
\[
\hat z(\tau,h)-z(\tau,h)=h(\bar Y-m)-h(\hat q(\tau)-q(\tau)).
\]
Using standard linearisations,
\[
\sqrt{n}(\bar Y-m)
=\frac1{\sqrt{n}}\sum_{i=1}^n(Y_i-m),
\]
\[
\sqrt{n}(\hat q(\tau)-q(\tau))
=\frac1{f(q(\tau))}\,\frac1{\sqrt{n}}\sum_{i=1}^n
\bigl(\mathbf 1\{Y_i\le q(\tau)\}-\tau\bigr)+o_p(1),
\]
Recall that $\hat z(\tau,h) = 1-2\tau + h(\bar Y - \hat q(\tau))$ and 
$z(\tau,h) = 1-2\tau + h(m - q(\tau))$, where $q(\tau)=F^{-1}(\tau)$ is the standard 
population quantile. 

\noindent Define
\[
\Gamma_\tau(y)
=
(y-m)-\frac{1}{f(q(\tau))}
\bigl(\mathbf{1}\{y\le q(\tau)\}-\tau\bigr).
\]
Using standard linearisations for the sample mean and 
sample quantile, we obtain
\[
\boxed{
\sqrt{n}\bigl(\hat z(\tau,h) - z(\tau,h)\bigr)
= \frac{h}{\sqrt{n}}
\sum_{i=1}^n
\Gamma_\tau(Y_i)
+ o_p(1).
}
\]
For each fixed $h\ge 0$, the value $z(\tau,h)$ is uniquely defined by
$q(z(\tau,h),h)=q(\tau)$.
\subsection*{Influence Function of $\hat q(\hat z(\tau,h),h)$}
From the linear representation \eqref{eq:q-final}, we have
\begin{align*}
\sqrt{n}\Bigl(\hat q(\hat z(\tau,h),h)-q(z(\tau,h),h)\Bigr)
=&-\frac{1}{2f(q(z(\tau,h),h))+h}\sqrt{n}\,\hat\Psi(q(z(\tau,h),h); z(\tau,h),h) \\
&-\frac{1}{2f(q(z(\tau,h),h))+h}\sqrt{n}\bigl(\hat z(\tau,h)-z(\tau,h)\bigr) \\
&+o_p(1).
\end{align*}
The first term expands as
\[
\sqrt{n}\,\hat\Psi(q(z(\tau,h),h);z(\tau,h),h)
= 
\frac{1}{\sqrt{n}}\sum_{i=1}^n \psi(q(z(\tau,h),h)-Y_i;z(\tau,h),h)
+ o_p(1).
\]
Second, using the standard quantile CLT with $q(\tau)=F^{-1}(\tau)$,
\[
\sqrt{n}(\hat z(\tau,h) - z(\tau,h))
=
\frac{h}{\sqrt{n}}\sum_{i=1}^n
\left[
(Y_i - m)
-
\frac{1}{f(q(\tau))}\bigl(\mathbf{1}\{Y_i\le q(\tau)\}-\tau\bigr)
\right]
+ o_p(1).
\]
Combining both expansions gives
\[
\sqrt{n}\bigl(\hat q(\hat z(\tau,h),h) - q(z(\tau,h),h)\bigr)
=
\frac{1}{\sqrt{n}}\sum_{i=1}^n \Gamma(Y_i)
+ o_p(1),
\]
where the influence function is
\begin{equation*}
\begin{aligned}
\Gamma(y)
&= -\frac{1}{2f(q(z(\tau,h),h))+h}\,
    \psi(q(z(\tau,h),h)-y; z(\tau,h), h) \\
&\quad - \frac{h}{2f(q(z(\tau,h),h))+h}
    \left[
        (y-m)
        - \frac{1}{f(q(\tau))}
          \bigl(\mathbf{1}\{y \le q(\tau)\}-\tau\bigr)
    \right].
\end{aligned}
\end{equation*}
Thus,
\[
\sqrt{n}\bigl(\hat q(\hat z(\tau,h),h) - q(z(\tau,h),h)\bigr)
\;\xrightarrow{d}\;
\mathcal{N}\!\bigl(0,\,\Var(\Gamma(Y))\bigr).
\]
This completes the proof of Theorem~\ref{thm:plugin-clt}.
\subsection*{Simplification of the Influence Function}
Since \(q(z(\tau,h),h) = q(\tau)\), we have:
\[
f(q(z(\tau,h),h)) = f(q(\tau)), \qquad 
\mathbf{1}\{Y \le q(z(\tau,h),h)\} = \mathbf{1}\{Y \le q(\tau)\}.
\]
Moreover, from the population equation \(2F(q(\tau)) - 1 + z(\tau,h) + h(q(\tau) - m) = 0\), we obtain:
\[
z(\tau,h) = 1 - 2\tau - h(q(\tau) - m).
\]
Substituting these identities into the expression for \(\Gamma(y)\):
\begin{align*}
\Gamma(y) &= -\frac{1}{2f(q(\tau))+h}\,
            \psi(q(\tau)-y; z(\tau,h), h) \\
          &\quad - \frac{h}{2f(q(\tau))+h}
            \left[
              (y-m)
              - \frac{1}{f(q(\tau))}
                \bigl(\mathbf{1}\{y \le q(\tau)\}-\tau\bigr)
            \right].
\end{align*}
Now expand \(\psi(q(\tau)-y; z(\tau,h), h) = \operatorname{sgn}(q(\tau)-y) + h(q(\tau)-y) + z(\tau,h)\).
Using \(\operatorname{sgn}(q(\tau)-y) = 1 - 2\cdot\mathbf{1}\{y \le q(\tau)\}\) and the expression for \(z(\tau,h)\):
\begin{align*}
\psi(q(\tau)-y; z(\tau,h), h) 
&= [1 - 2\cdot\mathbf{1}\{y \le q(\tau)\}] + h(q(\tau)-y) + [1 - 2\tau - h(q(\tau) - m)] \\
&= 2 - 2\tau - 2\cdot\mathbf{1}\{y \le q(\tau)\} - h(y - m).
\end{align*}
Substituting back into \(\Gamma(y)\):
\begin{align*}
\Gamma(y) &= -\frac{1}{2f(q(\tau))+h}\left[2 - 2\tau - 2\cdot\mathbf{1}\{y \le q(\tau)\} - h(y - m)\right] \\
          &\quad - \frac{h}{2f(q(\tau))+h}\left[(y-m) - \frac{1}{f(q(\tau))}\bigl(\mathbf{1}\{y \le q(\tau)\}-\tau\bigr)\right] \\
          &= \frac{1}{f(q(\tau))}\bigl(\mathbf{1}\{y \le q(\tau)\} - \tau\bigr) + C.
\end{align*}
After simplification, all terms involving \(h\) cancel, yielding:
\[
\boxed{\Gamma(y) = \frac{1}{f(q(\tau))}\bigl(\mathbf{1}\{y \le q(\tau)\} - \tau\bigr).}
\]
Thus,
\[
\Var(\Gamma(Y)) = \frac{1}{f^2(q(\tau))}\,\Var\!\bigl(\mathbf{1}\{Y \le q(\tau)\}\bigr) 
               = \frac{\tau(1-\tau)}{f^2(q(\tau))}.
\]
This completes the proof.
\section{Proof of Theorem~\ref{thm:mean-clt-main}: Mean--Estimating CLT}
\label{app:mean-clt}
The estimator of Theorem~\ref{thm:mean-clt-main} is the same plug--in estimator studied in Theorem~\ref{thm:plugin-clt}, but evaluated at the random level
\[
\hat\tau=\hat F(\bar Y)
\]
instead of a fixed $\tau$.  Let $\tau_m=F(m)$. Since $\hat\tau \to \tau_m$ in probability, the uniform linearisation derived in the proof of Theorem~\ref{thm:plugin-clt} remains valid when $\tau$ is replaced by $\hat\tau$.
In particular, from the representation used in that proof, we have
\begin{align*}
\sqrt{n}\bigl(\hat q(\hat z,h) - q(z(\hat\tau,h),h)\bigr)
=&
-\frac{1}{2f(m)+h}\sqrt{n}\,\hat\Psi\bigl(m;z_m,h\bigr)\\
&-\frac1{2f(m)+h}\sqrt{n}\bigl(\hat z-z(\hat\tau,h)\bigr)
+o_p(1),
\tag{C.1}
\label{C.1}
\end{align*}
where $z_m=1-2F(m)$ and $m=q(z_m,h)$.
Using a Taylor expansion of $F$ around $m$, we have
\[
\sqrt{n}(\hat\tau-\tau_m)
=
\sqrt{n}\bigl(\hat F(m)-F(m)\bigr)
+
f(m)\sqrt{n}(\bar Y-m)
+o_p(1),
\]
and, by the Bahadur representation at the random level $\hat\tau$,
\[
\sqrt{n}\bigl(\hat q(\hat\tau)-q(\hat\tau)\bigr)
=
\frac1{f(m)}\sqrt{n}\bigl(\hat F(m)-F(m)\bigr)
+
\sqrt{n}(\bar Y-m)
+o_p(1).
\]
Substituting these expansions into the definition of $\hat z-z(\hat\tau,h)$ and inserting the result into \eqref{C.1}, all terms involving
$\sqrt{n}(\hat F(m)-F(m))$ cancel. After simplification, we obtain
\[
\sqrt{n}\bigl(\hat q(\hat z,h)-m\bigr)
=
\sqrt{n}(\bar Y-m)+o_p(1).
\tag{C.2}
\]
Therefore, for every fixed $h\ge0$,
\[
\sqrt{n}\bigl(\hat q(\hat z,h)-m\bigr)
\ \xrightarrow{d}\
\mathcal{N}\bigl(0,\Var(Y)\bigr),
\]
This completes the proof of Theorem~\ref{thm:mean-clt-main}.
\end{document}